\documentclass{IEEEtran}
\onecolumn
\usepackage[pdftex]{graphicx}
\usepackage{amsmath}
\usepackage{amssymb}
\usepackage{amsthm}

\usepackage{cancel}
\usepackage{tikz}
\newcommand{\hc}[1]{%
    \tikz[baseline=(tocancel.base)]{
        \node[inner sep=0pt,outer sep=0pt] (tocancel) {#1};
        \draw[black] (tocancel.south west) -- (tocancel.north east);
    }%
}
\newtheorem{theorem}{Theorem}
\newtheorem{corollary}{Corollary}
\newtheorem{proposition}{Proposition} 
\newtheorem{definition}{Definition}
\newtheorem{remark}{Remark}
\newtheorem{lemma}{Lemma }
\newtheorem{example}{Example}
\newcommand{\cc}{\mathsf{c}}
\newcommand{\RC}{\ensuremath{R_{\mathrm{CAPM}}}}
\begin{document}

\sloppy

\title{A Rate-Distortion Approach to Index Coding} 
\author{Sinem Unal and Aaron B. Wagner~
\thanks{S.~Unal and A.~B.~Wagner are with  Cornell University, School of Electrical \& Computer Engineering, Ithaca, NY 14853 USA (e-mail: su62@cornell.edu, wagner@ece.cornell.edu). This paper was presented in part at the IEEE Int. Symposium on Information Theory (ISIT), Istanbul, July 2013 and Information Theory and Applications Workshop (ITA), San Diego, February 2014.}
}

\maketitle

\begin{abstract}
We approach index coding as a special case of 
rate-distortion with multiple receivers, each
with some side information about the source.  Specifically,
using techniques developed for the rate-distortion problem, we 
provide two upper bounds and one lower bound on the optimal
index coding rate. The upper bounds involve specific choices
of the auxiliary random variables in the best existing scheme for
the rate-distortion problem.
The lower bound is based on a new lower bound for the 
general rate-distortion problem. The bounds are shown to coincide
for a number of (groupcast) index coding instances, including
all instances for which the number of decoders does not exceed three.
\end{abstract}

\section{Introduction}
  We consider a general version of the \emph{index coding} problem, 
  in which a single encoder observes a vector-valued source, the components  
 of which are i.i.d., uniformly-distributed, binary random variables.    
 There are several decoders, each of which has a subset of the source    
 components as side information, and seeks to losslessly reproduce     
 a disjoint subset of source components. The encoder must broadcast  
 a single message to all of the decoders, which allows all of them   
 to reproduce their desired source components.   
 This problem models media distribution over a rate-constrained  
 downlink in which some of the clients already have some of the  
 media to be distributed, which can happen if they cached    
 prior transmissions from the encoder. We seek to understand     
 what rate is required of the encoder's message when the encoder     
 may code over many i.i.d.\ realizations of the source vector.   
 Specifically, we seek estimates of the minimal rate that are    
 both efficiently computable and provably close to the true  
 minimal rate, at least under some conditions. 
     
The index coding problem has attracted considerable attention   
since it was introduced~(e.g.,~\cite{birk1, birk, shanmugam_local,shanmugam,blasiak,neely,Maleki,arbabjolfaei}), and the formulations
studied vary along at least three different axis. First, one
can impose structure on the demands and the side information.
Birk and Kol's paper~\cite{birk1} introducing the problem
focused on the case in which
each source bit is demanded by exactly one decoder, and if
Decoder $i$ has Decoder $j$'s demand as side information then
Decoder $j$ must have Decoder $i$'s demand as side information.
The problem instance can then be represented as an undirected
graph, with the nodes representing the source bits (or 
equivalently, the decoders) and the edges representing the
side information pattern. Slightly more generally, one
can relax the symmetry assumption to obtain a directed,
instead of undirected, graph. Shanmugam~\emph{et al.}~\cite{shanmugam}
call this \emph{unicast} index coding. We shall consider here 
the more general version in which
each decoder can demand any number of source components
and the demands may be overlapping among the decoders.
Shanmugam~\emph{et al.} call this \emph{groupcast} index coding;
we shall simply call it \emph{index coding}.

A second independent axis along which index coding formulations
vary is whether one allows for coding over time (vector codes)
or whether the code must operate on each time instant separately
(scalar codes). We shall focus on the former here, due to its
intrinsic importance and its connection with rate-distortion
theory. Third, and finally, some works require the decoders to reproduce
their demands with zero error~\cite{birk,shanmugam, blasiak, neely, arbab_14} while others require
that the block-error probability vanish~\cite{ Maleki, arbabjolfaei}. Yet
another possibility is to require that the bit-error probability
vanish. This paper shall focus on the latter two.

Irrespective of the formulation, most work on index coding
views the problem graph-theoretically~\cite{ shanmugam_local, shanmugam, arbab_14}, as in Birk and Kol's original paper.
 One can then lower and upper bound the optimal rate using   
 graph-theoretic quantities such as the independence number,   
 the clique-cover number, fractional clique-cover number (e.g.~\cite{blasiak}),
 the min-rank~\cite{birk} and others~\cite{shanmugam_local,shanmugam, arbab_14}. This approach has  
 proven to be successful for showing the utility of coding over blocks   
 for this problem~\cite{block}, and for showing the utility of 
  nonlinear codes~\cite{nonlinear}.
 Many of these graph-theoretic quantities are known to be NP-hard to     
 compute, however, and for the others there is no apparent polynomial-time   
 algorithm. Thus these bounds are only useful theoretically or when     
 numerically solving small examples. A noteworthy exception is Theorem 2     
 of Blasiak \emph{et al.}, which provides a polynomial-time-computable
 bound
 that is within a nontrivial factor of the optimal rate for arbitrary    
 instances. The factor in question is quite large, however.  
     
 We approach index coding as a special case of the problem of lossy  
 compression with a single encoder and multiple decoders, each with  
 side information. This more general problem was introduced by 
 Heegard-Berger~\cite{berger} (but see Kaspi~\cite{Kaspi:Absent})
 and is sometimes referred to as the \emph{Heegard-Berger problem.}  
 Index coding can be viewed as the special case in which the source,     
 at each time, is a vector of i.i.d.\ uniform bits, the side information     
 of each decoder consists of a subset of the source bits, and the    
 distortion measure for Decoder $i$ is the Hamming distortion    
 between the subset of the source bits that Decoder $i$ seeks    
 to reproduce and Decoder $i$'s reproduction of that subset.     
 We then consider the minimum rate possible so that all of   
 the decoders can achieve zero distortion.   
     
 Viewing the problem in this way allows us to apply tools from   
 network information theory, such as random coding techniques,   
 binning, the use of auxiliary random variables, etc. Using this
 approach, we prove two achievable bounds
 and an impossibility (or ``converse'') bound.
  Both of the achievable bounds
 are built upon the best known achievable bound for the 
   Heegard-Berger problem, which
is due to Timo~\emph{et al.}~\cite{timo}. The Timo~\emph{et al.} scheme
involves an optimization over the joint distribution of a 
large number of auxiliary random variables; we provide
two methods for selecting this distribution, the first
of which is polynomial-time computable but only yields
integer rate bounds, while the second is more complex but
can yield fractional rates.
The achievability
results in this paper are thus unusual in that the emphasis
is on algorithms for selecting the joint distribution of
auxiliary random variables rather than proving new coding
theorems \emph{per se}.
It is worth noting that the Timo~\emph{et al.} 
result is representative of many achievability results
in network information theory that take the form of
optimization problems over the joint distribution of 
auxiliary random variables (e.g.,~\cite{ElGamal}). The task of 
solving these optimization problems has received little 
attention in the literature.\footnote{Indeed, the prospect
that some of these ``single-letter'' optimization problems might 
be intrinsically hard to compute is intriguing and seemingly
unexplored (though see Arikan~\cite{Arikan:packet}).} 

Our impossibility result is related to the 
``degraded-same-marginals'' (DSM) impossibility result for
broadcast channels~\cite{vishwanath, vishwanathP}.
The idea is that the optimal
rate can be computed exactly when the source $S$ and the
side information variables $Y_1, \ldots, Y_m$ can be
coupled in such a way that 
$$
S \leftrightarrow Y_{\sigma(1)} \leftrightarrow Y_{\sigma(2)} \leftrightarrow
  \cdots \leftrightarrow Y_{\sigma(m)},
$$
meaning that the random variables form a Markov chain in this
order, where $\sigma(\cdot)$ is an arbitrary permutation~\cite{timo}.
We call such an instance one with \emph{degraded side information}.
One may then lower bound a given problem by providing (for example),
$Y_1$ to Decoder 2, $Y_1$ and $Y_2$ to Decoder 3, etc., to form
a degraded instance whose optimal rate is only lower than that
of the original problem. We provide a lower bound in this spirit
for the general Heegard-Berger problem
that improves somewhat on that obtained via a direct application
of the above technique. As such, we shall call it the DSM+ lower
bound. When applied to the index coding problem, the DSM+ bound
provides the same conclusion as a lower bound due to 
Blasiak \emph{et al.}, although under slightly weaker
hypotheses.

We use the DSM+ lower bound to show that our low-complexity
achievable bound
equals the optimal rate for any number of source components,
so long as the number of decoders does not exceed 
three.\footnote{Recall that we allow each decoder to demand more
than one source component and each source component to be demanded
by more than one decoder.} In fact, we show the more
general result that the achievable bound equals the optimal
rate for any number of source components and any number of
decoders so long as each source component is present as
side information at all of the decoders, none of the decoders,
all but one of the decoders, or all but two of the decoders.
It is apparent that every problem with three or fewer decoders
must be of this form. We also show that the achievable
bound coincides with the optimal rate when none of the 
source components are ``excess,'' a concept that plays
an important role in our achievable scheme and that shall be defined later.
Our low-complexity achievable scheme bears some resemblance
to the \emph{partition multicast} scheme of Tehrani, Dimakis, and 
Neely~\cite{neely}.
Although our scheme does not subsume \emph{partition
multicast} (which is NP-hard to compute \cite{neely}), we do show
that it is optimal in all explicit instances of the problem 
for which
Tehrani, Dimakis, and Neely show that partition multicast is 
optimal.\footnote{Tehrani~\emph{et al.} also show that partition 
multicast is optimal for the  implicitly-defined class of instances
for which clique cover is optimal.}

Although the paper is focused mainly on index coding, the
results herein also have some significance for the Heegard-Berger
problem. The DSM+ bound, mentioned earlier, is the best general
lower bound for this problem, and our conclusive results for
the index problem represent some of the few nondegraded 
instances of the Heegard-Berger problem for which the 
optimal rate is known (see~\cite{watanabe,timo:complementary,sinem_gauss,sgarro,
timo:gray,timo:conditional} others). This paper 
is also the first work that considers algorithms for
selecting the distribution of the auxiliary random
variables in the Timo~\emph{et al.} scheme.

As noted earlier, this work differs from much of the
literature on index coding by approaching the problem
as one of rate-distortion, or source coding. Some
recent works have also approached the problem as
one of channel coding~\cite{Maleki,arbabjolfaei}, and in particular,
interference alignment. One of the advantages of the
source coding approach espoused here is that it can
readily accommodate richer source models and 
distortion constraints, including sources with
memory, lossy reconstruction of analog sources,
etc. The very formulation of index coding presumes
that the sources have already been compressed down
to i.i.d.\ uniform bits. Thus the index coding is
``separated'' from the underlying compression,
when in fact there might be some advantage to
combining the two, a topic that we shall
consider in subsequent work.

This paper is outlined as follows. Section~\ref{sec:problem_defn} formulates the Heegard-Berger problem and Section~\ref{sec:lower_bound} provides the DSM+ lower bound for it. Section~\ref{sec:problem_defn_index} formulates the index coding problem. Section~\ref{sec:lower_bound_index} and~\ref{sec:upper_bound_index}  provide a lower bound and an upper bound for the problem respectively. Section~\ref{sec:first_scheme}  describes our first scheme of index coding, and Section~\ref{sec:optimality} provides several optimality results for this scheme, including our results for three decoders. Section~\ref{sec:second_scheme} describes our second scheme.

\section{Problem Definition}
\label{sec:problem_defn}
We begin by considering the general form of the Heegard-Berger problem,
as opposed to the index-coding problem in particular.
There is a single encoder with source ${S}$ and there are $m$ decoders.
Decoder $i$ has a side information ${Y_i}$ that in general depends
on ${S}$.
The encoder sends a message at rate $R$ to the decoders, and Decoder $i$ 
wishes to reconstruct the source with a given distortion constraint $D_i$.
The objective is to find the rate distortion tradeoff for this problem setup.
This is made precise via the following definitions.
\begin{definition} 
\label{defn:code}
An $(n,M, D_1, \ldots, D_m)$ \emph{code} consists of mappings
	\begin{align*}
	&f  : \mathcal{S}^{n}  \rightarrow  \{ 1,\ldots, M \} \\
	&g_1 : \{ 1,\ldots, M \}  \times   \mathcal{Y}_1^{n} \rightarrow  								\hat{\mathcal{S}}_1^{n} \\
	&g_2 : \{ 1,\ldots, M \}  \times \mathcal{Y}_2^{n}  \rightarrow  									\hat{\mathcal{S}}_2^{n}\\
	&\vdots \\
	&g_m : \{ 1,\ldots, M \}  \times \mathcal{Y}_m^{n} \rightarrow  								\hat{\mathcal{S}}_m^{n},
	\end{align*}
	\begin{align*}
	 \mathbb{E}\left[\frac{1}{n} 
	\sum_{k = 1}^{n} d({S}_k,{\hat{S}}_{(i)k})\right] \le D_i, \   \forall i \in [m]
	\end{align*}
where $\mathcal{S}$ denotes the source alphabet, $\mathcal{Y}_1,\ldots, \mathcal{Y}_m$ denote the side information alphabets at Decoder $1$ through $m$ and $\hat{\mathcal{S}}_1^{n},\ldots, \hat{\mathcal{S}}_m^{n}$ denote the reconstruction alphabets at Decoder $1$ through $m$ and $d(.,.) \in [0 \ \infty)$ denotes a distortion measure and $[m] = \{1,\ldots,m\}$. Lastly, we call $f$  the \emph{encoding function at the encoder} and $g_i$  the 
\emph{decoding function at Decoder $i$} where $i\in [m]$.
\end{definition}

\begin{definition} 
	A \emph{rate distortion pair} $(R,D)$, where $D=(D_1,\ldots ,D_m)$ is \emph{achievable} if for every $\epsilon > 0$, there exists an 
	$(n,M, D_1 +\epsilon,\ldots ,D_m + \epsilon )$ code such that $n^{-1}\log{M} 	\le R + \epsilon$. 
\end{definition}

\begin{definition} The \emph{rate distortion function} $R(D)$ is defined as
	\begin{align*}
	R(D) =\inf\{R | (R,D) \ \textrm{is} \ \textrm{achievable} \}.
	\end{align*}
\end{definition}

Finding a computable characterization $R(D)$ is a long-standing open 
problem in network information theory. Currently, such a characterization
is only available for a few special cases. Heegard and Berger themselves
\cite{berger} provided one when the side information at the decoders is 
degraded. Watanabe~\cite{watanabe} provided one for the case that the
source consists of two independent components, the distortion constraints
for both decoders are decoupled across the two components, and the side
information at the two decoders is degraded ``in mismatched order'' (see
\cite{watanabe} for the precise setup). Sgarro's result~\cite{sgarro}
implies a characterization for the problem in which two decoders both
wish to reproduce the source losslessly, without any assumption
on their side information (see also~\cite{timo:gray}). 
Timo~\emph{et al.}~\cite{timo:conditional} 
provide a characterization for the two-decoder case when 
one decoder's side information is ``conditionally less noisy''
than the other's and the weaker decoder seeks to losslessly
reproduce a deterministic function of the source. 
Timo~\emph{et al.}~\cite{timo:complementary} solve various two-decoder cases
in which the source consists of two components, say $(X,Y)$,
and one decoder has $X$ as side information and wants to
reconstruct $Y$ while the other has $Y$ as side information
and wants to reconstruct $X$.
The present authors determined the rate distortion region for the two-decoder problem with vector Gaussian sources and side information, subject to a constraint on the error covariance matrices at the two decoders~\cite{sinem_gauss}. Several (nondegraded) special cases in which both decoders
wish to losslessly reproduce a function of the source have been
solved by Laich and Wigger~\cite{laich}. Of course, several instances
of index coding that are not degraded have also been solved.

A general achievable result, i.e.,
an upper bound on $R(D)$, was provided by Heegard and Berger~\cite{berger},
which was corrected and extended by Timo~\emph{et al.}~\cite{timo}. 
We provide a computable lower bound on $R(D)$ for general instances
of the problem in this section. This lower bound will be used later 
in the paper to solve several index coding instances.

\section{Lower Bound for a Rate Distortion Function}
\label{sec:lower_bound}
We start our analysis by providing a lower bound to the general rate distortion problem.
\begin{theorem}[DSM+ Lower Bound]
 \label{theorem_1}
 Let the pmf's $P(S,Y_i)$ for all $i \in [m]$ be given.
$R(D)$ is lower bounded by
 \begin{align}
 \label{lower_general}
 &R_{DSM+}(D) = \max_{\sigma} \sup_{\bar{P}} \bar{R}_{\sigma}(D) 
\end{align}
where 
\begin{align}
\label{eq:Umin}
\bar{R}_{\sigma}(D) =\min_{U_1,\ldots,U_m} & \big{[} I({S}; U_{\sigma(1)}|{Y_{\sigma(1)}})
+ I( {S}; U_{\sigma(2)}| U_{\sigma(1)},{Y_{\sigma(1)}},{Y_{\sigma(2)}}) +\cdots
\\
&+ I( {S}; U_{\sigma(m)}| U_{\sigma(1)},\ldots,U_{\sigma(m-1)},{Y_{\sigma(1)}},\ldots,{Y_{\sigma(m)}}) \big{]} \notag
\end{align}
and
 \\
1) $\sigma(.)$ denotes a permutation on integers $[m]$
\\
2) $\bar{P}=\{ P(S,Y_1, \ldots, Y_m) | \sum_{\\
Y_j : j\neq i  }P(S,Y_1, \ldots, Y_m)$ $= P(S,Y_i), \forall i \in [m]\}$ 
\\
3) $(U_{1},\ldots,U_m)$ is jointly distributed with ${S}, {Y_1},\ldots,{Y_m}$  such that 
$$({Y_1},\ldots, {Y_m}) \leftrightarrow {S} \leftrightarrow (U_{1}, \ldots, U_{m})$$
and
\\
4) there exist functions $g_1, \ldots, g_m$ such that
\begin{equation}
\label{eq:gfun}
\mathbb{E}[d({S},g_{\sigma(i)}( U_{\sigma(i)}, Y_{\sigma(i)}))] \le D_{\sigma(i)}  \forall i \in [m],
\end{equation}
5) $|U_{\sigma(i)}| \le |\mathcal{S}|\prod^{i-1}_{j=1}|U_{\sigma(j)}| + (m +2- i)$ for all $i \in [m]$.
 \end{theorem}

The idea behind the proof was described in the introduction.
Note that since the optimal rate only depends on the source and
side information through the ``marginals''
$$
(S,Y_i) \quad i \in [m],
$$
we may couple the $Y_i$ variables to form a joint distribution
$\bar{P}(S,Y_1,\ldots,Y_m)$ as we please, leading to the inner
optimization in~(\ref{lower_general}). Also note that a
direct application of the DSM idea would yield the weaker
bound in which $\cup_{j \le i} Y_{\sigma(j)}$ appears as
an argument to $g_{\sigma(i)}$ in~(\ref{eq:gfun}).
  
 \begin{proof}[Proof of Theorem \ref{theorem_1}]
 Let $P(S,Y_i)$ for all $i \in [m]$ be given and let permutation $\sigma(i) = i$ for all $i \in [m]$.
Let $(R,D)$ be an achievable rate distortion pair and  $\epsilon > 0$. 
Then there exists a $(n,M,D_1 +\epsilon,\ldots,D_m +\epsilon)$ code for some $n$ such that $\log M \le n(R + \epsilon)$.  

We can write,
\begin{align}
n(R + \epsilon) &\ge H(J) \notag \\
\label{chain_eq1}
&\ge I({S^n}, {Y^n_1},\ldots, {Y^n_m}; J) 
\end{align}
where $J$ is the output of the encoder, ${Y^n_1} =({Y_{11}},\ldots,{Y_{1n}})$ (for the ease of notation we drop the parentheses around the index of the random variable unless it causes ambiguity) and  ${Y_{1\hc{i}}}$ denotes all ${Y^n_1}$
but ${Y_{1i}}$. Then if we apply the chain rule to  $I({S^n}, {Y^n_1},\ldots, {Y^n_m}; J)$,  right hand side of (\ref{chain_eq1}) equals
\begin{align}
\label{chain_eq2}
& I( {Y^n_1}; J) + I( {Y^n_2}; J|{Y^n_1})+ \cdots 
 + I( {Y^n_m}; J|{Y^n_1},\ldots,{Y^n_{m-1}})+I( {S^n}; J|{Y^n_1},\ldots,{Y^n_{m}}) 
 \\
&\ge  I( {Y^n_2}; J|{Y^n_1})+ \cdots + I( {Y^n_m}; J|{Y^n_1},\ldots,{Y^n_{m-1}}) +I( {S^n}; J|{Y^n_1},\ldots,{Y^n_{m}}) 
\notag \\
\label{first_eq}
&\overset{a}{\ge} \sum^{n}_{i=1} \big{[} I( {Y_{2i}}; J,{Y_{1\hc{i}}}|{Y_{1i}})+ \cdots 
+ I( {Y_{mi}}; J,{Y_{1\hc{i}}},\ldots,{Y_{(m-1)\hc{i}}}|{Y_{1i}},\ldots,{Y_{(m-1)i}})
+I( {S_i}; J,{Y_{1\hc{i}}},\ldots,{Y_{m\hc{i}}}|{Y_{1i}},\ldots,{Y_{mi}}) \big{]}
\\
&\overset{b}{=} \sum^{n}_{i=1} \big{[} I( {Y_{2i}}; J,{Y_{1\hc{i}}}|{Y_{1i}})+ \cdots 
+ I( {Y_{mi}}; J,{Y_{1\hc{i}}},\ldots,{Y_{(m-1)\hc{i}}}|{Y_{1i}},\ldots,{Y_{(m-1)i}})
 + I( {S_i}; J,{Y_{1\hc{i}}},\ldots,{Y_{(m-1)\hc{i}}}|{Y_{1i}},\ldots,{Y_{mi}})  
\notag \\
&\quad + I( {S_i}; {Y_{m\hc{i}}}|J,{Y_{1\hc{i}}},\ldots,{Y_{(m-1)\hc{i}}},{Y_{1i}},\ldots,{Y_{mi}}) \big{]}
\notag 
\end{align}
where $a$ is obtained by the chain rule and 
$b$ is due to the chain rule applied to the last term. When we combine the second-to-last and third-to-last term above, we get
\begin{align}
&n(R + \epsilon) 
\notag \\
&\ge \sum^{n}_{i=1} \big{[} I( {Y_{2i}}; J,{Y_{1\hc{i}}}|{Y_{1i}})+ \cdots 
+ I({Y_{(m-1)i}}; J,{Y_{1\hc{i}}},\ldots,{Y_{(m-2)\hc{i}}}|{Y_{1i}},\ldots,{Y_{(m-2)i}})  
\notag \\
&\quad +I( {S_i}, {Y_{mi}}; J,{Y_{1\hc{i}}},\ldots,{Y_{(m-1)\hc{i}}}|{Y_{1i}},\ldots,{Y_{(m-1)i}}) 
+ I( {S_i}; {Y_{m\hc{i}}}|J,{Y_{1\hc{i}}},\ldots,{Y_{(m-1)\hc{i}}},{Y_{1i}},\ldots,{Y_{mi}}) \big{]} 
\notag \\
&\ge \sum^{n}_{i=1} \big{[} I( {Y_{2i}}; J,{Y_{1\hc{i}}}|{Y_{1i}})+ \cdots 
+ I({Y_{(m-1)i}}; J,{Y_{1\hc{i}}},\ldots,{Y_{(m-2)\hc{i}}}|{Y_{1i}},\ldots,{Y_{(m-2)i}}) 
\notag \\
\label{second_eq}
&\quad + I( {S_i}; J,{Y_{1\hc{i}}},\ldots,{Y_{(m-1)\hc{i}}}|{Y_{1i}},\ldots,{Y_{(m-1)i}}) 
 + I( {S_i}; {Y_{m\hc{i}}}|J,{Y_{1\hc{i}}},\ldots,{Y_{(m-1)\hc{i}}},{Y_{1i}},\ldots,{Y_{mi}}) \big{]}.
\end{align}

Now we apply the chain rule on the second-to-last term, giving
\begin{align}
&n(R + \epsilon) 
\notag \\
&\ge \sum^{n}_{i=1} \big{[} I( {Y_{2i}}; J,{Y_{1\hc{i}}}|{Y_{1i}})+ \cdots 
+ I({Y_{(m-1)i}}; J,{Y_{1\hc{i}}},\ldots,{Y_{(m-2)\hc{i}}}|{Y_{1i}},\ldots,{Y_{(m-2)i}}) 
 \notag \\
& \quad +I( {S_i}; J,{Y_{1\hc{i}}},\ldots,{Y_{(m-2)\hc{i}}}|{Y_{1i}},\ldots,{Y_{(m-1)i}}) 
+I( {S_i}; {Y_{(m-1)\hc{i}}}|J,{Y_{1\hc{i}}},\ldots,{Y_{(m-2)\hc{i}}},{Y_{1i}},\ldots,{Y_{(m-1)i}}) 
 \notag \\
 \label{ineq:eqn}
& \quad + I( {S_i}; {Y_{m\hc{i}}}|J,{Y_{1\hc{i}}},\ldots,{Y_{(m-1)\hc{i}}},{Y_{1i}},\ldots,{Y_{mi}}) \big{]} 
\end{align}
\begin{align}
 &\overset{c}{=}\sum^{n}_{i=1} \big{[} I( {Y_{2i}}; J,{Y_{1\hc{i}}}|{Y_{1i}})+ \cdots 
+ I({Y_{(m-2)i}}; J,{Y_{1\hc{i}}},\ldots,{Y_{(m-3)\hc{i}}}|{Y_{1i}},\ldots,{Y_{(m-3)i}})  
\notag \\
& \phantom{\overset{d}{=}}+I( {S_i},{Y_{(m-1)i}}; J,{Y_{1\hc{i}}},\ldots,{Y_{(m-2)\hc{i}}}|{Y_{1i}},\ldots,{Y_{(m-2)i}}) 
+I( {S_i}; {Y_{(m-1)\hc{i}}}|J,{Y_{1\hc{i}}},\ldots,{Y_{(m-2)\hc{i}}},{Y_{1i}},\ldots,{Y_{(m-1)i}}) 
 \notag \\
&\phantom{\overset{d}{=}}+ I( {S_i}; {Y_{m\hc{i}}}|J,{Y_{1\hc{i}}},\ldots,{Y_{(m-1)\hc{i}}},{Y_{1i}},\ldots,{Y_{mi}}) \big{]} 
\notag \\
&\ge \sum^{n}_{i=1} \big{[} I( {Y_{2i}}; J,{Y_{1\hc{i}}}|{Y_{1i}})+ \cdots 
+ I({Y_{(m-2)i}}; J,{Y_{1\hc{i}}},\ldots,{Y_{(m-3)\hc{i}}}|{Y_{1i}},\ldots,{Y_{(m-3)i}}) 
\notag \\
&\phantom{\overset{d}{=}}+I( {S_i}; J,{Y_{1\hc{i}}},\ldots,{Y_{(m-2)\hc{i}}}|{Y_{1i}},\ldots,{Y_{(m-2)i}}) 
+I( {S_i}; {Y_{(m-1)\hc{i}}}|J,{Y_{1\hc{i}}},\ldots,{Y_{(m-2)\hc{i}}},{Y_{1i}},\ldots,{Y_{(m-1)i}}) 
\notag \\
 \label{third_eq}
&\phantom{\overset{d}{=}}+ I( {S_i}; {Y_{m\hc{i}}}|J,{Y_{1\hc{i}}},\ldots,{Y_{(m-1)\hc{i}}},{Y_{1i}},\ldots,{Y_{mi}}) \big{]} 
\end{align}
where, $c$ is obtained by  combining third-to-last and fourth-to-last terms in (\ref{ineq:eqn}).
Note that (\ref{second_eq}) is obtained from (\ref{first_eq}) by applying a series of chain rules and term combinations.  If we continue this procedure as we did while obtaining (\ref{third_eq}) from  (\ref{second_eq}), we get

\begin{align}
&R + \epsilon  \notag \\
&\ge \frac{1}{n}\sum^{n}_{i=1} \big{[} I( {S_{i}}; J,{Y_{1\hc{i}}}|{Y_{1i}})+ 
I( {S_i}; {Y_{2\hc{i}}}|J,{Y_{1\hc{i}}},{Y_{1i}},{Y_{2i}}) +\cdots
\label{fourth_eq}
+ I( {S_i}; {Y_{m\hc{i}}}|J,{Y_{1\hc{i}}},\ldots,{Y_{(m-1)\hc{i}}},{Y_{1i}},\ldots,{Y_{mi}}) \big{]} 
 \\
&\overset{a}{=} \frac{1}{n}\sum^{n}_{i=1} \big{[} I( {S_{i}}; U_{1i}|{Y_{1i}})+ 
I( {S_i}; U_{2i}| U_{1i},{Y_{1i}},{Y_{2i}}) +
 \cdots+ I( {S_i}; U_{mi}| U_{1i},\ldots,U_{(m-1)i},{Y_{1i}},\ldots,{Y_{mi}}) \big{]}
\notag \\
&\overset{b}{=} \frac{1}{n}\sum^{n}_{i=1} \big{[} I( {S_{i}}; U_{1i}|{Y_{1i}}, T=i)
 + I( {S_i}; U_{2i}| U_{1i},{Y_{1i}},{Y_{2i}}, T=i) +\cdots 
+ I( {S_i}; U_{mi}| U_{1i},\ldots,U_{(m-1)i},{Y_{1i}},\ldots,{Y_{mi}}, T=i) \big{]} 
\notag \\
\label{eq:lower1}
&\overset{c}{=}  I( {S}; U_{1}|{Y_{1}})+ 
I( {S}; U_{2}| U_{1},{Y_{1}},{Y_{2}}) +\cdots
+ I( {S}; U_{m}| U_{1},\ldots,U_{(m-1)},{Y_{1}},\ldots,{Y_{m}})
\end{align}
 where,
 \\ 
a : $U_{ji} = (J, Y_{j\hc{i}})$, $\forall j \in [m]$ 
\\
b : $T$ is a random variable uniformly distributed on $[m]$ and independent of all source and side information variables and the $U_{ji}$s.  
\\
c: Relabel $(U_i, T)$ as $U_i$ for all $i \in [m]$.
\\

Note that $(U_1, \ldots, U_m)$ satisfies 
conditions 3) and 4). By Lemma \ref{lemma:cardinality} in the Appendix, we can obtain the cardinality bounds on $(U_1,\ldots, U_m)$ as in condition 5).
Then we minimize the right hand side of (\ref{eq:lower1}) over $(U_1,\ldots, U_m)$.  Since the problem can be described  by specifying only the marginal $P(S,Y_i)$'s, we optimize it over the set of joint distributions such that  the marginal $P(S,Y_i)$'s are the same. This gives us a  lower bound to $R(D)$.
Lastly note that, we fixed the permutation as $\sigma(i) =i$ for all $i \in [m]$ and to get (\ref{chain_eq2}), we applied the chain rule to (\ref{chain_eq1})
in the following order. We started with ${Y^n_{1}}$ then continued
with ${Y^n_{2}},\ldots,{Y^n_{m}}$ and lastly we had ${S^n}$.  Since we have $m$ decoders with side information, we can get $m!$ different permutations.  Hence, applying a similar procedure to all permutations, we get $m!$ lower bounds. By taking their maximum, we obtain a lower bound, $R_{DSM+}(D+\epsilon\mathbf{1})$, where $\mathbf{1}$ denotes the $m\times 1$ vector with all components $1$.

Hence we have,
\begin{equation}
\label{eq:lower_epsilon}
R(D) \geq  R_{DSM+}(D+\epsilon\mathbf{1}) -\epsilon. 
\end{equation}
\begin{lemma}
\label{lemma:continuity}
$R_{DSM+}(D + \epsilon\mathbf{1})$ of Theorem \ref{theorem_1} is continuous in $\epsilon$ from the right at $\epsilon = 0$.
\end{lemma}
\begin{proof}[Proof of Lemma \ref{lemma:continuity}]
Since finite maximum of  functions that are continuous from the right is also continuous from the right, it is enough to prove that $\sup_{\bar{P}}\bar{R}_{\sigma}(D+\epsilon\mathbf{1})$ is continuous in $\epsilon$ from right.

First we show that for a given joint distribution $P(S,Y_1,\ldots,Y_m)$, $\bar{R}_{\sigma}(D+\epsilon\mathbf{1})$ is continuous in $\epsilon$ from right for a given permutation $\sigma(.)$.
\\
Let  $\epsilon_k$ be a monotonically decreasing sequence converging to $0$
and let $ U_1(D+\epsilon_k\mathbf{1}),\ldots, U_m(D+\epsilon_k\mathbf{1})$ denote an optimal $(U_1,\ldots,U_m)$ which gives $\bar{R}_{\sigma}(D+\epsilon_k \mathbf{1})$. 
Since the cardinalities of $(U_1,\ldots,U_m)$  are finite, together with the conditions 3) and 4),
 we have an optimization over a compact set. Then,  we can find a convergent subsequence $\epsilon_{s_k}$ such that 
$U_1(D+\epsilon_{s_k} \mathbf{1}),\ldots, U_m(D+\epsilon_{s_k} \mathbf{1})$ converges to a $(U_1,\ldots,U_m)$ which is feasible for the distortion $D$. Hence we have
\begin{align*}
\liminf_{\epsilon \rightarrow 0} \bar{R}_{\sigma}(D + \epsilon\mathbf{1}) \ge \bar{R}_{\sigma}(D).
\end{align*}

Also, since $\bar{R}_{\sigma}(D )$ is non increasing function of $D$ we can write 
\begin{align*}
\limsup_{\epsilon \rightarrow 0} \bar{R}_{\sigma}(D + \epsilon\mathbf{1}) \le \bar{R}_{\sigma}(D),
\end{align*}
concluding that $\bar{R}_{\sigma}(D+\epsilon\mathbf{1})$ is continuous in $\epsilon$ from the right for a given permutation $\sigma(.)$.

Now we show that $\sup_{\bar{P}}\bar{R}_{\sigma}(D+\epsilon\mathbf{1})$ is continuous in $\epsilon$ from the right.
Let us temporarily write $\bar{R}_{\sigma}(D)$ as $\bar{R}_{\sigma}(\bar{P},D)$ to indicate the dependence on $\bar{P}$.
Then for any $\epsilon > 0$ and any $\bar{P}$, we have
$\bar{R}_{\sigma}(\bar{P},D + \epsilon\mathbf{1}) \le \bar{R}_{\sigma}(\bar{P},D)$. Hence we have
\begin{align}
\label{eq:max_1}
\limsup_{\epsilon \rightarrow 0} \sup_{\bar{P}} \bar{R}_{\sigma}(\bar{P},D + \epsilon\mathbf{1}) \le \sup_{\bar{P}} \bar{R}_{\sigma}(\bar{P},D).
\end{align}

Now, we fix any $\delta >0$. Then there exists $P'$ such that
$\bar{R}_{\sigma}(P',D)\ge \sup_{\bar{P}} \bar{R}_{\sigma}(\bar{P},D) - \frac{\delta}{2}$. Let $\epsilon >0$ be such that  
$\bar{R}_{\sigma}(P',D + \epsilon\mathbf{1} )\ge \bar{R}_{\sigma}(P',D ) - \frac{\delta}{2}$. Then for any $0 < \epsilon'< \epsilon$ we can write
\begin{align*}
 \sup_{\bar{P}} \bar{R}_{\sigma}(\bar{P},D)& \le \bar{R}_{\sigma}(P',D) +\frac{\delta}{2}
 \\
 &\le \bar{R}_{\sigma}(P',D + \epsilon'\mathbf{1} ) + \delta
 \\
  &\le \sup_{\bar{P}}\bar{R}_{\sigma}(\bar{P},D + \epsilon'\mathbf{1} ) + \delta,
\end{align*}
implying that 
\begin{align}
\label{eq:max_2}
 \sup_{\bar{P}} \bar{R}_{\sigma}(\bar{P},D) - \delta  \le \liminf_{\epsilon' \rightarrow 0}\sup_{\bar{P}}\bar{R}_{\sigma}(\bar{P},D + \epsilon'\mathbf{1} ).
\end{align}
Since $\delta > 0$ was arbitrary, (\ref{eq:max_1}) and (\ref{eq:max_2}) give the result.
\end{proof}
By Lemma \ref{lemma:continuity}, when $\epsilon$ goes to $0$, (\ref{eq:lower_epsilon}) becomes

\begin{align}
\label{lowerb_1}
R(D) \geq  R_{DSM+}(D).
\end{align}
\end{proof}

\begin{remark}
\label{remark:minmax_lb}

Since the proof constructs the $U_i$ variables in a way that
does not depend on the permutation $\sigma(\cdot)$ or the
coupling $\bar{P}$, one could state the bound with the
minimum over $U_1,\ldots,U_m$ outside the maximum over
$\sigma(\cdot)$ and $\bar{P}$. This complicates the
proof of the cardinality bounds in 5), however, and the
maxmin form of the bound is sufficient for the purposes
of this paper, so we shall defer consideration of this
potential strengthening to a later work.
\end{remark}
Next, we  turn our focus to the \textit{index coding} problem which can be viewed as a special case of the Heegard-Berger problem.
\section{index coding : Problem Formulation}
\label{sec:problem_defn_index}
 For the $m$ user index coding problem, each decoder $\alpha$ wants to reconstruct 
 $\mathbf{f_\alpha}$, which is an arbitrary subset of the source $\mathbf{S}$, that is, a collection of i.i.d. Bernoulli($\frac{1}{2}$) bits at the encoder. There may be \textit{overlapping demands}, i.e., more than one decoder may demand the same bit.  Also, each Decoder $\alpha$ has side information $\mathbf{Y_\alpha}$ consisting of an arbitrary subset of the source.  We assume that decoders do not demand a component of their own side information since they already have it, and we assume that $\mathbf{Y_\alpha} \neq \mathbf{Y_\beta}$, for all $\alpha \neq \beta$ since we can combine two decoders if they have the same side information. We may also assume that every source bit is demanded by at least one decoder, for otherwise that bit may be completely purged from the system.

Let $\mathbf{S_{J}}$ denote the part of the source which each decoder in a subset $J$ of $[m]$ does not have 
and all decoders in $[m]\backslash J$ have as side information. 
If $J=\{ \alpha\}$, i.e., a singleton, then for  ease of notation we use $\mathbf{S_\alpha}$ instead of $\mathbf{S_{\{\alpha\}}}$. Since there are $m$ decoders, we group the elements of $\mathbf{S}$ into $2^m$ disjoint sets such that $\mathbf{S}=\cup_{J\subseteq[m]}\mathbf{S_J}$.
Note that  each $\mathbf{S_J}$ may be empty, may consist of a single bit, or may consist of multiple bits. 

 Let $G_0=\mathbf{S_{[m]}}$ denote the elements of the source that none of the decoders have, $G_m=\mathbf{S_\emptyset}$ denote the elements all decoders have, $G_{m-1}=\cup_{\alpha \in [m]}\mathbf{S_\alpha}$ denote elements that $m-1$ of the decoders have, $G_{m-2}=\cup_{
  \begin{tiny}
 \begin{aligned}
 &\{\alpha,\beta\} \subseteq [m]\\
 &\alpha \neq \beta
 \end{aligned}
  \end{tiny}
 }\mathbf{S_{\{\alpha,\beta\}}}$ denote elements that $m-2$ of the decoders have and so on. To ease the notation for the rest of the paper,  whenever we write a set $\{\alpha,\beta \}$, we assume $\alpha\neq \beta$ unless otherwise stated. Then $\mathbf{S}$ can be represented as $\mathbf{S}=\{G_0,G_m,G_{m-1},\ldots,G_1\}$, as shown in Fig.~\ref{fig:setup}.  

\begin{figure}[t]
  \begin{center}
 \includegraphics[scale=0.35]{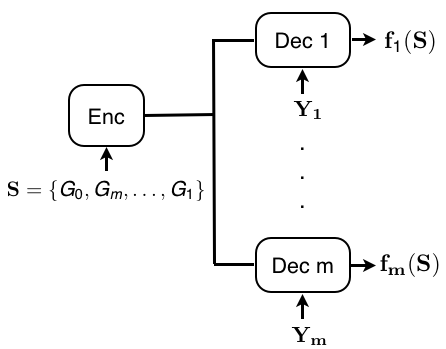}
\caption{Index coding with $m$ users.}
   \label{fig:setup}
 \end{center}
\end{figure}

The demand $\mathbf{f_\alpha}$ at Decoder $\alpha$ can be written in terms of components $\mathbf{S_J}$ of $\mathbf{S}$. For this, we introduce the following notation.

 Let $f_{IJ}$ denote the demand that is a subset of source $\mathbf{S_J}$ and is required by each decoder in a subset $I$ of $[m]$ and by no decoders in $[m]\backslash I$. If $I=\{\alpha\}$, then for ease of notation we use $f_{\alpha J}$ instead of $f_{\{\alpha\} J}$.
We will generally assume that $I \subseteq J$ since only decoders in $J$ may have a demand about $\mathbf{S_J}$ and decoders in $[m] \backslash J$ already have $\mathbf{S_J}$ as side information. If $I \not\subseteq J$, $f_{IJ}$ is empty.  Also, $f_{IJ}$ and $f_{KJ}$ are independent (i.e., $f_{IJ} \perp f_{KJ}$) for all possible choices of $I$, $K$ and $J$ with $I \neq K$ since $f_{IJ} \cap f_{KJ}= \emptyset$ unless $I=K$. Lastly, each $f_{IJ}$ may be empty, a single bit or may consist of multiple bits.

We have written the source as 
$\mathbf{S} =\{G_0,G_m,G_{m-1}$,
$\ldots,G_1\}$ and the demands in terms of $\mathbf{S_J}$'s. From now on, we consider an ordered set structure on $\mathbf{S}$ which naturally induces orders on $\mathbf{S_J}$'s. Then each demand $f_{IJ}$ is also an ordered set that can also be viewed as a vector. In fact, we shall find it convenient to view $f_{IJ}$, $\mathbf{S}$, and other similar quantities at times as sets and at times as vectors.

Since this problem can be considered as a special case of the Heegard-Berger problem, we use a similar definition for the code except for the distortion.  Specifically, we consider block error probabilities instead of the distortion constraints stated in Definition \ref{defn:code}. Hence, we use the following definitions for the code, error and optimal rate.

\begin{definition} Let $\mathcal{S}$ denote the alphabet of $\mathbf{S}$. An $(n,M)$ \emph{code} consists of mappings
\begin{align*}
&f  : \mathcal{S}^{n}  \rightarrow  \{ 1,\ldots, M \} \\
&g_1 : \{ 1,\ldots, M \}  \times  \mathcal{Y}_1^{n} \rightarrow  \mathcal{F}_1^{n}  \\
&g_2 : \{ 1,\ldots, M \}  \times  \mathcal{Y}_2^{n} \rightarrow  \mathcal{F}_2^{n} \\
&\vdots \\
&g_m : \{ 1,\ldots, M \}  \times \mathcal{Y}_m^{n} \rightarrow  \mathcal{F}_m^{n} ,
\end{align*}
where we  $f$  denotes the encoding function at the encoder and $g_\alpha$  denotes the 
decoding function at Decoder $\alpha$ where $\alpha \in [m]$.
\end{definition}

 \begin{definition} The \emph{probability of error} for a given code is defined as
\begin{align*}
P_e= Pr\{& g_1(f(\mathbf{S^n}),\mathbf{Y_1^n}){\neq}\mathbf{f_1^n}(\mathbf{S^n}) 
\cup {g_2}(f(\mathbf{S^n}),\mathbf{Y_2^n}){\neq}\mathbf{f_2^n}(\mathbf{S^n})
,\ldots,
\cup  g_m(f(\mathbf{S^n}),\mathbf{Y_m^n}){\neq}\mathbf{f_m^n}(\mathbf{S^n})\}.
\end{align*}
\end{definition}

Then achievability and optimal rate can be defined as follows.
\begin{definition} 
The rate $R$ is \emph{achievable} if 
there exists a sequence of  
$(n,M)$ codes with rate $n^{-1}\log M \le R$ such that the probability of 
error, $P_e$, tends to zero as $n$ tends to infinity.
\end{definition}

 \begin{definition} The \emph{optimal rate} $R_{opt}$ is defined as
\begin{align*}
R_{opt} =\inf\{R | R \ \textrm{is} \ \textrm{achievable} \}.
\end{align*}
\end{definition}

We shall call the problem defined in this section \emph{index coding},
although most existing work on index coding requires the code to 
achieve zero, as opposed to vanishing, block error~\cite{shanmugam,  blasiak,neely}. In support
of the definitions adopted here, see~\cite{Maleki,arbabjolfaei} for works that
use vanishing block error probability and~\cite{langberg} for
results connecting the two formulations.

\section{Lower Bound for Index Coding}
\label{sec:lower_bound_index}
The next theorem gives a lower bound to the index coding problem using the DSM+ lower bound from Section \ref{sec:lower_bound}.
\begin{theorem}
\label{thm:lower_opt}
The optimal rate of the index coding problem is lower bounded by
\begin{align}
R_{DSM+} & = \max_{\sigma}\big[ H(\mathbf{f_{\sigma(1)}}|\mathbf{Y_{\sigma(1)}})
+ H( \mathbf{f_{\sigma(2)}}|\mathbf{f_{\sigma(1)}},\mathbf{Y_{\sigma(1)}},\mathbf{Y_{\sigma(2)}}) 
+\cdots+
\label{eqn:index_low}
H( \mathbf{f_{\sigma(m)}}|\mathbf{f_{\sigma(1)}},\ldots,\mathbf{f_{\sigma(m-1)}},\mathbf{Y_{\sigma(1)}},\ldots,\mathbf{Y_{\sigma(m)}}) \big]
\end{align}
where $\sigma(.)$ denotes a permutation on integers $[m]$.
\end{theorem}

\begin{proof}[Proof of Theorem \ref{thm:lower_opt}]

We will use the lower bound in Theorem \ref{lower_general} to prove the theorem.
Note that this lower bound is for per-letter distortion constraints but it can be adapted to handle block error probabilities in the following way.  Vanishing error probability, $P_e$, for index coding problem implies vanishing block error probability for each Decoder $i$, i.e., $Pr(g_i(f(\mathbf{S^n}),\mathbf{Y_i^n}){\neq}\mathbf{f_i^n}(\mathbf{S^n}) )$, which implies vanishing distortion with respect to Hamming distortion measure for Decoder $i$. Also, note that lower bound in Theorem \ref{lower_general} is continuous from right by Lemma \ref{lemma:continuity}. Hence, the optimal rate for the index coding problem, $R_{opt}$, is lower bounded by
\begin{align*}
R_{opt} \ge \max_{\sigma} \min_{U_1,\ldots,U_m}  &\big{[} I(\mathbf{S}; U_{\sigma(1)}|\mathbf{Y_{\sigma(1)}})
+ I( \mathbf{S}; U_{\sigma(2)}| U_{\sigma(1)},\mathbf{Y_{\sigma(1)}},\mathbf{Y_{\sigma(2)}}) +\cdots
\notag \\
&+ I( \mathbf{S}; U_{\sigma(m)}| U_{\sigma(1)},\ldots,U_{\sigma(m-1)},\mathbf{Y_{\sigma(1)}},\ldots,\mathbf{Y_{\sigma(m)}}) \big{]}
\end{align*}
such that
 \\
1) $\sigma(.)$ denotes a permutation on integers $[m]$
\\
2) $(U_{1},\ldots,U_m)$ jointly distributed with $\mathbf{S}, \mathbf{Y_1},\ldots,\mathbf{Y_m}$,  and 
\begin{align}
\label{cond:decode_index}
&H(\mathbf{f_{\sigma(i)}}| U_{\sigma(i)}, \mathbf{Y_{\sigma(i)}}) = 0, \forall i \in [m].
\end{align}

Note that the maximum over $\bar{P}$ in (\ref{lower_general}) is degenerate 
for index coding. Without loss of generality let $\sigma(i) =i$ for all $i \in [m]$. Then we have, 
\begin{align}
\label{ineq:lb_index}
&R_{opt} \ge \min_{U_1,\ldots,U_m}  \big{[} I(\mathbf{S}; U_{1}|\mathbf{Y_{1}})
+ I( \mathbf{S}; U_{2}| U_{1},\mathbf{Y_{1}},\mathbf{Y_{2}}) +\cdots
+ I( \mathbf{S}; U_{m}| U_{1},\ldots,U_{m-1},\mathbf{Y_{1}},\ldots,\mathbf{Y_{m}}) \big{]}
\end{align}

To find an explicit expression for (\ref{ineq:lb_index}), we use the following lemma.
\begin{lemma}
\label{lemma:lb_explicit}
For $j \in [m]$ we define 
\begin{align*}
K_j =& \sum^{j-1}_{i=1} I(\mathbf{f_i}; U_1^i | \mathbf{Y_1^i}, \mathbf{f_1^{i-1}}) + I(\mathbf{S}; U_1^j | \mathbf{Y_1^j}, \mathbf{f_1^{j-1}}) 
 +\sum^{m}_{i=j+1} I(\mathbf{S}; U_i | \mathbf{Y_1^i}, U_1^{i-1}),
\end{align*}
where $\mathbf{f^{i}_1} = (\mathbf{f_1},\ldots,\mathbf{f_i})$ and likewise for $U^i_1$ etc.
Then $K_1 \ge K_2 \ge \ldots \ge K_m$.
\end{lemma}
\begin{proof}[Proof of Lemma \ref{lemma:lb_explicit}]
We fix any $j \in [m-1]$ and write,
\begin{align*}
&K_j - K_{j+1} 
\\
&= - I(\mathbf{f_j}; U_1^j | \mathbf{Y_1^j}, \mathbf{f_1^{j-1}}) + I(\mathbf{S}; U_1^j | \mathbf{Y_1^j}, \mathbf{f_1^{j-1}}) 
 -I(\mathbf{S}; U_1^{j+1} | \mathbf{Y_1^{j+1}}, \mathbf{f_1^{j}}) + I(\mathbf{S}; U_{j+1} | \mathbf{Y_1^{j+1}}, U_1^{j})
\\
&\overset{a}{=}  I(\mathbf{S}; U_1^j | \mathbf{Y_1^j}, \mathbf{f_1^{j}})  -I(\mathbf{S}; U_1^{j+1} | \mathbf{Y_1^{j+1}}, \mathbf{f_1^{j}})
 + I(\mathbf{S}; U_{j+1} | \mathbf{Y_1^{j+1}}, U_1^{j})
\\
&{=}  I(\mathbf{S}; U_1^j | \mathbf{Y_1^j}, \mathbf{f_1^{j}})  -I(\mathbf{S}; U_1^{j} | \mathbf{Y_1^{j+1}}, \mathbf{f_1^{j}})
-I(\mathbf{S}; U_{j+1} | \mathbf{Y_1^{j+1}}, \mathbf{f_1^{j}},U_1^{j})
 + I(\mathbf{S}; U_{j+1} | \mathbf{Y_1^{j+1}}, U_1^{j})
\\
&\overset{b}{\ge} 0,
\end{align*}
where
\\
a: is due to the chain rule.
\\
b: is due to the  side information and reconstructions being subsets of the source, $\mathbf{S}$.
\end{proof}
Then (\ref{ineq:lb_index}) becomes,
\begin{align}
&R_{opt}
\notag \\
&\ge K_1
\notag \\
&\ge K_m, \mbox{ by Lemma \ref{lemma:lb_explicit}}
\notag \\ 
&= \sum^{m-1}_{i=1} I(\mathbf{f_i}; U_1^i | \mathbf{Y_1^i}, \mathbf{f_1^{i-1}}) + I(\mathbf{S}; U_1^m | \mathbf{Y_1^m}, \mathbf{f_1^{m-1}})
\notag \\
&\ge \sum^{m}_{i=1} I(\mathbf{f_i}; U_1^i | \mathbf{Y_1^i}, \mathbf{f_1^{i-1}})
\notag \\
&= \sum^{m}_{i=1} H(\mathbf{f_i} | \mathbf{Y_1^i}, \mathbf{f_1^{i-1}}) - H(\mathbf{f_i}| U_1^i , \mathbf{Y_1^i}, \mathbf{f_1^{i-1}})
\notag \\
\label{rlo1}
&= \sum^{m}_{i=1} H(\mathbf{f_i} | \mathbf{Y_1^i}, \mathbf{f_1^{i-1}}),  \mbox{ from (\ref{cond:decode_index})}.
\end{align}
Applying the same procedure to all $m!$ permutations gives the result.
\end{proof}

\begin{remark}
Evidently the proof shows that the conclusion holds even if one 
 only requires that the bit-error probability, as opposed to the
block-error probability, vanish.
\end{remark}

\begin{remark}
Let us consider one of the $m!$ expressions of the lower bound in Theorem \ref{thm:lower_opt}, say the one in (\ref{rlo1}). We can rewrite it as
\begin{align}
\label{eqn:lb_index}
& H(\mathbf{f_1}\setminus \mathbf{Y_1}) + H(\mathbf{f_2}\setminus\{\mathbf{f_1}\cup \mathbf{Y_1} \cup \mathbf{Y_2} \}) + \cdots +
H(\mathbf{f_m} \setminus \{ \mathbf{Y_m}, \cup^{m-1}_{i =1} \{\mathbf{Y_i}, \mathbf{f_i}\}\}) 
  \\
  &= |\mathbf{f_1}\setminus \mathbf{Y_1}| + |\mathbf{f_2}\setminus\{\mathbf{f_1}\cup \mathbf{Y_1} \cup \mathbf{Y_2} \}| + \cdots +
 |\mathbf{f_m} \setminus \{ \mathbf{Y_m}, \cup^{m-1}_{i =1} \{\mathbf{Y_i}, \mathbf{f_i}\}\}|.
  \notag
\end{align}
Blasiak~\emph{et al.}~\cite{blasiak} define an \emph{expanding sequence}
of decoders as one for which each decoder in the sequence demands
a bit that is not contained in the union of the demands and the
side information of the decoders that appear earlier in the sequence.
Blasiak~\emph{et al.} prove that the size of a largest expanding
sequence is a lower bound on the optimal rate. Writing the above
bound as in (\ref{eqn:lb_index}) shows that it coincides with the 
Blasiak~\emph{et al.} bound when each decoder demands a single
bit. Of course, the more general case in which a decoder may
demand multiple bits can be obtained from the Blasiak~\emph{et al.}
result by replacing each such decoder with multiple decoders that
each demand a single bit.  The Blasiak~\emph{et al.} result 
does not quite imply Theorem~\ref{thm:lower_opt}, however, 
since the former assumes a zero-error formulation (though one
could appeal to a result of Langberg and Effros~\cite{langberg} 
to relate the two formulations).
\end{remark}
\section{Achievable Scheme for Index Coding}
\label{sec:upper_bound_index}

For our achievable scheme for index coding, we rely on an achievability
result of Timo~\emph{et al.}~\cite{timo}, mentioned earlier, for the general 
Heegard-Berger problem (see also Heegard and Berger~\cite{berger}). Since
the Timo~\emph{et al.} scheme is rather complicated, we shall state
it in a substantially weakened form that will be sufficient for our
purposes.

\begin{proposition}{(cf.~\cite[Theorem 2]{timo})}
\footnote{The weakened form of this result that we stated in the conference version of this work~\cite{sinemITA} was invalid in that it omitted one of the necessary conditions on the auxiliary random variables. The result has been corrected here.}
\label{ach_general}
The optimal rate $R_{opt}$ of an index coding problem is upper bounded by
\begin{align}
\label{eqn:ach_index1}
\min \sum_{I\subseteq [m]}  \left[ \max_{
i\in I}  H(U_{I} |\mathbf{{Y_i}}) \right]
\end{align}
where the minimization is over the set of all random variables $U_{I}$ jointly distributed with $\mathbf{S}$ such that
\\
1) There exist functions 

$g_1(\cup_{1 \in I}U_I,\mathbf{Y_1})$,$\ldots$, $g_m(\cup_{m \in I}U_I,\mathbf{Y_m})$ 
  such that
\begin{align*}
& g_i(\cup_{i \in I}U_I,\mathbf{Y_i}) = \mathbf{f_i}(\mathbf{S}), \mbox{ for all } i \in [m].
\end{align*}
2) The auxiliary random variables $U_I$, $I \subseteq [m]$ are independent,
and for all collections of subsets $J_1, \ldots J_j$, $K_1, \ldots K_k$,
$L_1, \ldots, L_l$, and all subsets $\{i_1, \ldots, i_p\} \subseteq [m]$,
we have that
$(U_{J_1}, \ldots, U_{J_j})$ and $(U_{K_1}, \ldots, U_{K_k})$ are
conditionally independent given $((U_{L_1}, \ldots, U_{L_l}), 
(\mathbf{Y_{i_1},\ldots, Y_{i_p}})$, provided that the collections
 $J_1, \ldots J_j$ and  $K_1, \ldots K_k$ are disjoint.
 \\
3) Each $U_{I}$ is a (possibly empty) vector of bits, each of which 
is the mod-2 sum of a set (possibly singleton) of source components.
\end{proposition}

The full-strength version of Timo~\emph{et al.}'s result omits
conditions 2) and 3) but replaces the rate expression in (\ref{eqn:ach_index1})
with one that is more complex. Under conditions 2) and
3), however, their expression reduces to (\ref{eqn:ach_index1}).
Also Timo~\emph{et al.} state their result as an upper bound on
$R(D)$ defined in Section~\ref{sec:problem_defn}, as opposed to $R_{opt}$ as
defined in Section~\ref{sec:problem_defn_index}. That is, they provide a guarantee
on the expected time-average distortion, instead of on the block
error probability that we use to define index coding. Their proof
technique can be used to bound the block error probability with
minimal modification, however.

One way of interpreting $U_I$ is that it is a ``message'' that
is ``sent'' to all Decoders $i$ such that $i \in I$. That is,
$U_I$ includes some information about the source that is 
decoded by all of the decoders in $I$ but is not available
to any of the decoders in $I^c$. The contribution of $U_I$
to the overall rate in (\ref{eqn:ach_index1}) is simply the rate needed
to send $U_I$ to all of the decoders in $I$ using standard
binning arguments (and relying on the fact that $U_I$ is
a deterministic function of the source $\mathbf{S}$).

Evaluating this upper bound requires finding the optimal
joint distribution of the $U_I$ auxiliary random variables.
Since each $U_I$ is a deterministic function of $\mathbf{S}$,
this is equivalent to finding the optimal such functions.
Such an optimization problem is evidently quite complicated.
We shall provide a polynomial-time heuristic for finding
a feasible choice of the $U_I$s. Of the many different
index coding schemes that have been proposed (e.g.,~\cite{ shanmugam,birk, neely, Maleki,arbabjolfaei} ),
ours most closely resembles the \emph{partition multicast}
 of Tehrani, Dimakis, and 
Neely~\cite{neely}. In the language of our setup, their scheme
amounts to finding the optimal choice of the $U_I$ subject
to the constraint that each $U_I$ must be a vector 
consisting of a (possibly empty) subset of the source 
components. Tehrani~\emph{et al.} show that finding this
optimal choice is NP-hard~\cite{neely}. Our scheme, in
contrast, consists of three steps, the first two of which
amount to a polynomial-time heuristic for finding 
a reasonable and feasible (but not necessarily optimal)
choice of auxiliary random variables subject to the
constraint that each $U_I$ must be a vector 
consisting of a subset of the source components. 
Thus the output of the second step of our heuristic
is a feasible solution to the optimization problem
for which partition multicast is optimal. Our third
step, however, replaces some of the $U_I$ variables
with ones that are more general functions of source,
i.e., not just subsets of the source variables. 
Due to the similarity between our heuristic and 
partition multicast, we call our heuristic
\emph{coded approximate partition multicast (CAPM)}.
Although CAPM is not guaranteed to be
never worse than partition multicast, we shall
show that it is optimal for all of the
explicit scenarios for which Tehrani~\emph{et al.} show
that partition multicast is optimal as well as
some other, more general scenarios.
\section{CAPM: Selection of $U_{I}$'s in the Achievable Scheme for Index Coding}
\label{sec:first_scheme}
CAPM is a method for choosing a feasible choice of 
the auxiliary random variables $U_{I}$ for $I \subseteq [m]$.  
Note that the number
of auxiliary random variables is exponential in the number of
decoders, although in typical instances most of these random
variables will be null. To minimize the worst-case complexity
of CAPM, therefore, we shall work with a linked list
of the auxiliary random variables that are not null, which
shall begin empty.
We shall call all $U_{I}$ auxiliary random variables for which
$|I| = i$ ``\textit{level $i$ messages}." 

\textbf{Step 1} :  
Beginning with an empty linked list of auxiliary random variables, we 
sequence through the vector of source bits. Any given bit must be in
$f_{KJ}$ for some $K \subseteq J \subseteq [m]$. So long as $J \ne [m]$,
we seek to include
this bit in $U_{K \cup J^\cc}$: if $U_{K \cup J^\cc}$ does not
exist in our linked list of auxiliary random variables, then we
add it to the list and set it equal to the source bit in question.
If it already exists in the list, then we locate it, and 
we set $U_{K \cup J^\cc}$ to be a vector of bits consisting
of all source bits that were included previously along with this
newly included source bit. For a source bit in $f_{KJ}$ where
$J = [m]$, we include the bit in the auxiliary random variable
$U_{[m]}$, i.e., the auxiliary random variable that is decoded
by all of the decoders. This process is repeated until all of
the source bits have been included in an auxiliary random variable.
Note that each nonvoid auxiliary random variable is then simply a vector
of source bits. Also note that each source bit will be included
in exactly one auxiliary random variable. 

We now sort the linked list so that all level-2 messages appear
first, followed by all level-3 messages, etc. Note that all level-1
messages are necessarily empty (assuming there is more than 
one decoder), by virtue of the fact that every
source bit is assumed to be demanded by at least one decoder,
and
source bits that no decoder has as side information
are placed in $U_{[m]}$.
The complexity of Step 1 is at most $O(s^2 \cdot m)$, where
$s = |\mathbf{S}|$. 
\begin{remark}
See Proposition~\ref{theorem:index} to follow for a justification of this 
particular approach to allocating the source components among the
different auxiliary random variables.
\end{remark}

\textbf{Step 2} : 
Let $U_I$ denote the first auxiliary random variable in the linked list.
If $I = [m]$, i.e., this first auxiliary random variable is decoded by
all of the decoders, then this $U_I$ must be the only nonnull auxiliary
random variable (since they are sorted by level), in which case we
skip Step 2 and proceed to Step 3. Suppose instead that $|I| < m$.
Note that $U_I$'s contribution to the overall rate is
$$
\max_{i \in I} H(U_I|\mathbf{Y_i}).
$$

In many cases $H(U_I|\mathbf{Y_i})$ will not be constant over $i \in I$. That is, some decoders in $I$ will require a higher rate to decode
$U_I$ than others. When this happens we move some of the source
bits in $U_I$ to a higher-level message. Define the two decoder indices
\begin{equation}
\label{eq:step2min}
i^* = \min\{i: H(U_I|\mathbf{Y_i}) = \min_{l \in I} H(U_I|\mathbf{Y_l})\}
\end{equation}
and
\begin{equation}
\label{eq:step2max}
j^* = \min\{j: H(U_I|\mathbf{Y_j}) = \max_{l \in I} H(U_I|\mathbf{Y_l})\}.
\end{equation}

If $H(U_I|\mathbf{Y_{i^*}}) < H(U_I|\mathbf{Y_{j^*}})$, then there must exist a source
bit in $U_I$ that is contained in $\mathbf{Y_{i^*}}$ but not
in $\mathbf{Y_{j^*}}$. We select the lowest-index source bit with this property
and move it from $U_I$ to some $U_J$ such that $I \subset J$
and $|J| = |I| + 1$. If $|I| < m-1$, then there are many such
choices of $J$; $J$ can be chosen arbitrarily, but for concreteness
we shall assume the following. First we look for nonempty $U_J$'s such that $I \subset J$
and $|J| = |I| + 1$. If we can find such a message or messages, we select the $J$ with the lowest index that is not already in $I$. If that is not the case, $J$ is obtained by adding to $I$ the lowest index that is not already in $I$.  We call the bit
that is moved \emph{leftover} or \emph{excess.} We then recompute
$i^*$ and $j^*$ according to (\ref{eq:step2min}) and 
(\ref{eq:step2max}), respectively, and move an additional
bit to a higher-level message if necessary, repeating this
process until $U_I$ is such that $H(U_I|\mathbf{Y_{i^*}}) = H(U_I|\mathbf{Y_{j^*}})$.
Note that this condition must eventually be satisfied, since
after sufficiently many iterations, $U_I$ will become null.
Once this condition is satisfied for $U_I$, we apply the
same procedure to the next auxiliary random variable in the
linked list, and so on until this procedure has been applied
to every variable in the linked list. It is possible that
some auxiliary random variables in the linked list are made
null through this procedure, in which case they are removed
from the linked list. The complexity of Step 2 is $O(m^2 \cdot s^3)$.

\begin{remark}
The rationale for moving source bits up to higher-level messages
is as follows. A bit that is excess contributes to the maximum
\begin{align}
\label{eqn:max_entropy}
\max_{l \in I} H(U_I|\mathbf{Y_l}),
\end{align}
which is $U_I$'s contribution to the overall rate. Thus removing
this bit from $U_I$ has the potential to reduce $U_I$'s contribution
to the rate (although it will not necessarily do so, if there
are multiple $l$ that achieve the maximum in (\ref{eqn:max_entropy}); see the
next remark). Of course, including this bit in a higher-level
message, $U_J$, will tend to increase $U_J$'s contribution to
the rate. But it will only do so the source bit in question
is not in the side information of one of the decoders $l$ that
achieve the maximum in
$$
\max_{l \in J} H(U_I|\mathbf{Y_J}).
$$

Thus moving the bit up one level often yields a rate reduction, and 
even if it does not, it may yield a rate reduction upon being
elevated again during a later iteration.
\end{remark}

\begin{remark}
If there exists a unique $j \in I$ such that $H(U_I|\mathbf{Y_j}) = \max_l H(U_I|\mathbf{Y_l})$,
then moving an excess bit to a higher-level message cannot increase
the overall rate, and in some cases it may strictly decrease the rate.
If the decoder with maximum rate is not unique, then moving an excess
bit to a higher-level message can increase the rate, as in Example~\ref{example:CAPM}
to follow, although this increase is sometimes offset during later
movements of excess bits, or during Step 3 (again as in Example~\ref{example:CAPM}).
For this reason we move excess bits according to the procedure outlined
in Step 2 even when such movements have the immediate effect of
increasing the overall rate.
\end{remark}

\begin{remark}
Finding the feasible allocation of source components among
the various $U_I$ variables that minimizes the rate in (\ref{eqn:ach_index1})
is NP-hard, as shown by Tehrani~\emph{et al.}~\cite{neely}.
\end{remark}

\textbf{Step 3} : In the final step, we exclusive-OR (XOR) some of the
bits included in the auxiliary random variables. Let $U_I$ denote
the first auxiliary random variable in the linked list, and suppose
that $V_1, \ldots, V_l$ denote the excess source bits that are included
in $U_I$. Recall that bits placed in $U_{[m]}$ during Step 1 are
not considered excess.
For each $i$, let $N_i$ denote the set of decoders that
need (i.e., demand) $V_i$ and let $H_i$ denote the set of decoders
that have $V_i$ has side information. We search for a pair of
components $V_i$ and $V_j$ such that $N_i \subset H_j$,
$N_j \subset H_i$, and $V_i$ and $V_j$ were included in the
same auxiliary random variable in Step 1 (that is, $N_i \cup H_i
= N_j \cup H_j$). If there are no such $V_i$ and $V_j$ then we
proceed to the next $U$ variable in the linked list. Otherwise,
we delete $V_j$ from $U_I$,
we replace $V_i$ in $U_I$ with $V_i \oplus V_j$, we replace
$N_i$ with $N_i \cup N_j$ and $H_i$ with $H_i \cap H_j$. Since
both $V_i$ and $V_j$ were placed in the same auxiliary random
variable in Step 1, we view the new $V_i$ as also being placed 
in that variable in Step 1, although of course
the auxiliary random variables constructed in Step 1 did not involve taking
the XOR of any of the source components.  We then repeat this 
process, again looking for $V_i$ and $V_j$ such that 
$N_i \subset H_j$, $N_j \subset H_i$, and $V_i$ and $V_j$ 
were included in the same auxiliary random variable in Step 1.
If we find such a pair, we replace them with their exclusive-OR.
We repeat this process until there are no such pairs remaining.
We then apply this procedure to all of the other auxiliary random
variables in the linked list. The complexity of Step 3 is
$O(m\cdot s^3)$.

\begin{remark}
Evidently Step 3 will never increase the rate.
\end{remark}

\begin{remark}
One could certainly exclusive-OR bits $V_i$ and $V_j$ 
satisfying $N_i \subset H_j$ and $N_j \subset H_i$ but
for which $V_i$ and $V_j$ are not included in the same
auxiliary random variable in Step 1 or for which either
$V_i$ or $V_j$ are not excess bits. Choosing to exclusive-OR
certain pairs of bits can foreclose other such choices,
however, and the latter choices may ultimately lead to
lower rates. The restriction that we only exclusive-OR
bits that are excess and that originated in the same 
auxiliary random variable in Step 1 is intended to 
guide the process toward the most productive exclusive-OR
choices. Of course, once the above process exhausts all
of its exclusive-OR possibilities, one could look for
exclusive-OR opportunities among bits that are not excess
or that did not originate in the same auxiliary random
variable. We shall not include this step in the heuristic,
however, since it is not necessary in any of our optimality
results or any of our examples.

\end{remark}

One can verify that this selection procedure provides a feasible 
choice of the $U_I$ variables as follows. First note that
the choice will be feasible
after each step 1. This is because each source component is included
in a $U_I$ variable that is decoded by all of the decoders that
demand it. Thus condition 1) in 
Proposition~\ref{ach_general} is satisfied. Conditions 2) and 3)
are satisfied because each $U_I$ consists of a subset of
the source components and these subsets are disjoint across
the $U_I$s. Step 2 only moves source components
from a $U_I$ to a $U_J$ for which $I \subseteq J$, so it
is evident that conditions 1)-3) continue to hold.
Finally, the exclusive-OR operation applied in Step 3
evidently never violates conditions 2) or 3), and the
specific conditions under which the exclusive-OR operation
is applied ensures that condition 1) continues to hold.

\begin{definition}
The achievable rate provided by CAPM is denoted by $\RC$.
\end{definition}

To illustrate CAPM, we provide two examples.

\begin{example}
\label{example:CAPM}
Consider the $4$-decoder index coding problem instance with demands $f_{12^\cc},f_{32^\cc},f_{3\{1,2\}^\cc},f_{1\{2,3\}^\cc},f_{4[4]},f_{2[4]}$ where  each demand is one bit  and $a^\cc = [m] \setminus \{a\}$.
Now we show each step of CAPM.

\textit{Step $1$}:  At the end of this step we have
\\
$U_{12} = f_{12^\cc},  U_{23} = f_{32^\cc}$, 
\\
$U_{123} = f_{3\{1,2\}^\cc},f_{1\{2,3\}^\cc}$,
\\
$ U_{1234} = f_{4[4]},f_{2[4]}$.

\textit{Step $2$}: We start with level-$2$ messages. The first level-$2$ message is $U_{12}$.  $f_{12^\cc}$ in $U_{12}$ is an excess bit and since we already have level-$3$ message $U_{123}$ which $f_{12^\cc}$ can be placed
we move $f_{12^\cc}$ to $U_{123}$. The next message is $U_{23}$. $f_{32^\cc}$ in $U_{23}$ is an excess bit and it is also placed to $U_{123}$.
The messages at this point are
\\
 $U_{123} = f_{3\{1,2\}^\cc},f_{1\{2,3\}^\cc}, f_{12^\cc}, f_{32^\cc}$,
 \\
$ U_{1234} = f_{4[4]},f_{2[4]}$
\\
and  we move on to level-$3$ messages.
Note that there is only one level-$3$ message, $U_{123}$. 
All demands in it are excess bits since $H(U_{123}| \mathbf{Y_2})=0$. We move all excess bits to $U_{1234}$, which is the only one level-$4$ message. This completes the Step $2$ and we have 
\\
$U_{1234} =  f_{4[4]}, f_{2[4]} ,f_{3\{1,2\}^\cc},f_{1\{2,3\}^\cc}, f_{12^\cc}, f_{32^\cc}$
\\
at the end of this step.

\textit{Step $3$}: 
Note that $(f_{3\{1,2\}^\cc},f_{1\{2,3\}^\cc})$ are the only excess bits that were in the same message at Step $1$ and $f_{3\{1,2\}^\cc} \oplus f_{1\{2,3\}^\cc}$ is decodable at the respective decoders. Hence at the end of Step $3$, selection of the messages is the following:
\\
$U_{1234} = f_{4[4]}, f_{2[4]} ,f_{3\{1,2\}^\cc}\oplus f_{1\{2,3\}^\cc}, f_{12^\cc}, f_{32^\cc}$ 
\\and all others are empty.

Note that after Step $2$ the total rate is $6$ bits, whereas after Step $3$ the rate is $5$ bits. The lower bound in 
Theorem \ref{thm:lower_opt} also gives $5$ bits, showing that CAPM achieves the optimal rate for this example. 
\end{example}

\begin{remark}
\label{remark2:heuristic}
Let $R_{a}$ be a rate obtained by placing  $f_{IJ}$'s in messages ($U_{K}$'s)  by applying  Step $1$. Let  $R^*_{a}$ be a rate obtained such that $f_{IJ}$'s are placed in messages ($U^*_{K}$'s) by following the Step $1$, and applying the Step $2$ only for level- $2$ messages.

Note that $U_{K} = U^*_{K}$, for all level-$i$ messages where $i >3$. Also, each possible excess bit (or bits) which we will denote by $f^*_{IJ} \subseteq f_{IJ}$, coming from a level-$2$ message $U^*_{K}$ is such that either $I = \{\alpha\}$ or $I=\{\beta \}$ where $K =\{\alpha, \beta\}$. 
Then, when level-$2$ messages $U_{K}$ and $U^*_{K}$, where $K = \{\alpha, \beta \}$, are not the same, we can write $U^*_{K} \cup f^*_{IJ} = U_{K}$ and
\begin{align}
&\max_{i \in K} \{ H(U_{K}| \mathbf{Y_i})\}=\max_{i \in K} \{ H(U^*_{K} \cup f^*_{IJ}| \mathbf{Y_i})\} \notag \\
&= \max_{i \in K} \{ H(U^*_{K}| \mathbf{Y_i}) + H (f^*_{IJ}| \mathbf{Y_i})\} \notag \\
&\overset{a}{=}  H(U^*_{K}| \mathbf{Y_i}) +\max_{i \in K} \{ H (f^*_{IJ}| \mathbf{Y_i})\} \notag \\
\label{eq:remark2}
&\overset{b}{=}  H(U^*_{K}| \mathbf{Y_i}) + H(f^*_{IJ}) 
\end{align}
a: Since all $H(U^*_{K}| \mathbf{Y_i})$ for $i \in K$ are the same.
\\
b: Since $f^*_{IJ}$ is such that either $I = \{\alpha\}$ or $I=\{\beta \}$.

Hence, we can write $R_a - R^*_a$ as
\begin{align*}
&\left(\sum_{|I|=3}\max_{i \in I} \{ H(U_{I}| \mathbf{Y_i})\}-
\sum_{|I|=3}\max_{i \in I} \{ H(U^*_{I}| \mathbf{Y_i})\}
\right) +
\left(\sum_{|K|=2}\max_{i \in K} \{ H(U_{K}| \mathbf{Y_i})\}-\sum_{|K|=2}\max_{i \in K} \{ H(U^*_{K}| \mathbf{Y_i})\}\right)
\\
=&\left(\sum_{|I|=3}\max_{i \in I} \{ H(U_{I}| \mathbf{Y_i})\}-
\sum_{|I|=3}\max_{i \in I} \{ H(U^*_{I}| \mathbf{Y_i})\}
\right) 
+\sum_{f^*_{IJ}}H(f^*_{IJ}), \quad \mbox{ from (\ref{eq:remark2})}.
\end{align*}

Since $U^*_{I} = U_{I} \cup f^{*}_{I}$ where $f^{*}_{I}$ denotes all of the excess bits in $U^*_{I}$ and 
\begin{align*}
\max_{i \in I} \{ H(U_{I} \cup f^*_{I}| \mathbf{Y_i})\}
\le \max_{i \in I} \{ H(U_{I} | \mathbf{Y_i})\} +H( f^*_{I}),
\end{align*}
we can write
\begin{align*}
&R_a - R^*_a \ge -\sum_{|I|=3}H(f^*_{I})+\sum_{f^*_{IJ}}H(f^*_{IJ}) =0.
\end{align*}
\end{remark}
Note that since we apply Step $2$ once to the level-$2$ messages, the leftover bits that we get are unique, i.e., independent of the  
 sorted demand sequence given at the beginning of Step $1$ and 
 different leftover bits coming from previous levels (since there is no leftover bit coming to level-$2$ messages).
However, if we apply CAPM for an arbitrary instance of an index coding problem,  this may not be the case. In other words, at Step $2$ of the CAPM, we may get different excess bits due to differently sorted demand sequence given at the beginning of Step $1$ or different leftover bits coming from previous levels and this may affect the resulting rate. Also, when there are multiple options for leftover bits to be moved,  one may get different rates due to the selection of different next level messages to move the leftover bits. Lastly, there may be instances of index coding problem where not moving the bits to the next level gives  a lower rate. 
 To illustrate some of these issues, we provide the following examples. For the examples,  $f^a_{IJ} \backslash U_{K}$ denotes the leftover bits ($a$ bits) of $f_{IJ}$ from $U_{K}$. If all $f_{IJ}$ are leftover bits then we remove the superscript $a$.
 \begin{example}
\label{ex:heuristic1}
Consider the $4$-decoder index coding problem instance with demands $f_{12^\cc},f_{23^\cc},f_{31^\cc},f_{3\{1,2\}^\cc}$, $f_{4[4]}$ where $f_{12^\cc},f_{23^\cc}$ are two bits and the rest are one bit.
Now, we explain each step of the CAPM for this example.

\textit{Step 1:} At the end of this step we have
\\
 $ U_{12} = f_{12^\cc}$, 
$U_{23} = f_{23^\cc}$, $ U_{13} = f_{31^\cc}$, 
\\
$U_{123}=f_{3\{1,2\}^\cc}$, $U_{1234}=f_{4[4]}$.

\textit{Step 2:} We begin with level-$2$ messages. Note that all demands at level-$2$ messages are excess bits and can be moved to  level-$3$ message, $U_{123}$. Then we have,
\\
$U_{123}=f_{3\{1,2\}^\cc},  f_{12^\cc}, f_{23^\cc}, f_{31^\cc}$,
\\
$U_{1234}=f_{4[4]}$.

We move on to level-$3$ messages. There is only one level-$3$ message, $U_{123}$, and one bit of  $f_{12^\cc}$, denoted as $f^1_{12^\cc}$,  is an excess bit. Then we move it to $U_{1234}$, concluding Step $2$. Hence, messages at the end of this step are
\\
$U_{123}=f_{3\{1,2\}^\cc}, f^1_{12^\cc}, f_{23^\cc}, f_{31^\cc}$, 
\\
$U_{1234}=f_{4[4]}, f^1_{12^\cc}$.

\textit{Step 3:} Since there is no $\oplus$ opportunity as described in Step $3$, the messages at the end of Step $2$ remains the same, 
giving a total rate of $5$ bits.

Note that without loss of generality we can label Decoder $3$ as $1$ and Decoder $1$ as $3$. Then, if we apply CAPM with this relabeling we get the following messages at each step.

\textit{Step 1}: At the end of this step we have
\\
 $U_{32} = f_{32^\cc}$, 
$U_{12} = f_{21^\cc}$, $ U_{13} = f_{13^\cc}$, 
\\
$U_{123}=f_{1\{2,3\}^\cc}$, $U_{1234}=f_{4[4]}$.

\textit{Step 2:} We begin with level-$2$ messages. Similar to previous case all demands at level-$2$ messages are excess bits and moved to $U_{123}$. Then we have,
\\
$U_{123}=f_{1\{2,3\}^\cc},  f_{32^\cc}, f_{21^\cc}, f_{13^\cc}$,
\\
$U_{1234}=f_{4[4]}$.

We move on to level-$3$ messages. As in the previous case, there is only one level-$3$ message, $U_{123}$. However, now the excess bits of $U_{123}$ are 
$f_{1\{2,3\}^\cc}, f^1_{21^\cc}, f^1_{32^\cc}$. Then we move these to $U_{1234}$, concluding Step $2$. Hence, messages at the end of this step are
\\
$U_{123}= f^1_{32^\cc}, f^1_{21^\cc}, f_{13^\cc}$, 
\\
$U_{1234}=f_{4[4]}, f_{1\{2,3\}^\cc}, f^1_{21^\cc}, f^1_{32^\cc},$

\textit{Step 3:} Since there is no $\oplus$ opportunity as described in Step $3$, the messages at the end of Step $2$ remains the same, 
giving a total rate of $6$ bits. Thus rate achieved by
the heuristic depends on the indexing of the decoders.
\end{example}
\begin{example} 
\label{example:CAPM_n}
Consider the $5$-decoder index coding problem instance with demands $f_{1[5]},f_{5[5]},f_{2\{1,4\}^\cc},f_{3\{1,2\}^\cc},f_{4\{1,3\}^\cc}$ where  each demand is one bit.
Now we show each step of the CAPM.

\textit{Step 1:} At the end of this step we have,
\\
$U_{124} = f_{2\{1,4\}^\cc}$, $U_{123} = f_{3\{1,2\}^\cc}$,
$U_{134}= f_{4\{1,3\}^\cc}$,
\\
$U_{12345} = f_{1[5]},f_{5[5]}$.

\textit{Step 2:} We begin with the lowest level, i.e., level-$3$ for this example. Note that all of the demands in level-$3$ messages are excess bits and they are moved to $U_{1234}$. Then we have
\\
$U_{1234} =f_{2\{1,4\}^\cc},f_{3\{1,2\}^\cc},f_{4\{1,3\}^\cc}$,
\\
$U_{12345} = f_{1[5]},f_{5[5]}$.
\\
Note that total rate is $4$ bits  at this state.
 We move on to level-$4$ messages. There is only one level-$4$ message, $U_{1234}$. Since $H(U_{1234} |\mathbf{Y_1}) = 0$, all demands in $U_{1234}$ are excess bits and they are moved to $U_{12345}$. Then we have 
 \\
 $U_{12345} = f_{1[5]},f_{5[5]}, f_{2\{1,4\}^\cc},f_{3\{1,2\}^\cc}, f_{4\{1,3\}^\cc}$,
 where the total rate is $5$ bits.

\textit{Step 3:}
Since there is no $\oplus$ opportunity as described in Step $3$, the messages at the end of Step $2$ remains the same and total rate is $5$ bits. Note that if we did not move the excess bits at level-$4$ message, the total rate would be $4$ bits.
\end{example}
In the next section, we show that applying CAPM gives us the optimal rate for several specific cases of the index coding problem. 
\section{Optimality Results for Index Coding}
\label{sec:optimality}
We shall show that CAPM yields the optimal rate for
several scenarios. Since the partition multicast scheme of
Tehrani~\emph{et al.}~\cite{neely} is the most direct antecedent of 
CAPM, we begin by showing that CAPM coincides
with the DSM+ lower bound, and is thus optimal, for all
of the explicit scenarios for which Tehrani~\emph{et al.} show that
partition multicast is optimal.

\begin{figure}[t]
  \begin{center}
 \includegraphics[scale=0.3]{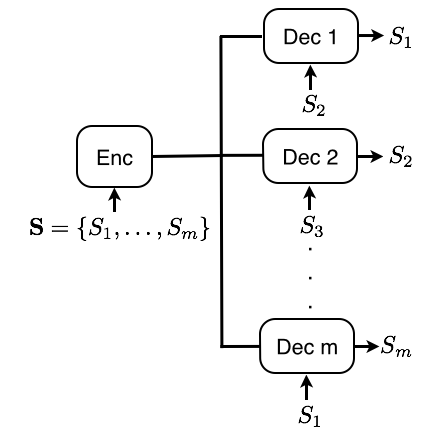}
\caption{Index coding instance which is also called "directed cycle"}
   \label{fig:ex_neely}
 \end{center}
\end{figure}
First, consider the case depicted in Fig.~\ref{fig:ex_neely},
in which there are $m$ decoders, $m$ source bits, and Decoder $k$
demands source bit $k$ and has source bit $k + 1$ as side information,
for $k \in \{1,\ldots,m-1\}$. Decoder $m$ demands source bit $m$ and
has the first source bit as side information. Such an instance is
typically called a ``directed cycle'' after its graph-theoretic
description.

\begin{proposition}
\label{prop:directedcycle}
For the instance depicted in Fig.~\ref{fig:ex_neely}, the achievable
rate provided by CAPM and the lower bound provided
by the DSM+ bound coincide. In fact
$$
\RC = R_{DSM+} = m-1.
$$
\end{proposition}
\begin{proof}
First we show that $\RC =m-1$. After Step $1$ of CAPM, the messages 
are $U_{\{i,i+1\}}= S_{i+1}$ for all $i \in [m-1]$ and 
$U_{\{1,m\}} = S_1$.
Observe that at any point of the algorithm, for any non-empty message $U_I$ where $|U_I| = k$, there exists $\mathbf{Y_j}\in I$ such that $H(U_I|\mathbf{Y_j})= k-1$ and $k \ge H(U_I|\mathbf{Y_i})\ge k-1$ for all $i \in I$. Hence after Step $2$ of CAPM, 
$H(U_I|\mathbf{Y_j})$ 
will be equal to  $k-1$ for all $j \in I$ for any nonempty message $U_I$.
Now we show that for any non-empty message $U_I$ with $|U_I| = k$,
$H(U_I|\mathbf{Y_j}) = k-1$ for all $j \in I$ if and only if 
$|U_I| = m$. This will imply that $\RC \le m-1$. Consider a $U_I$
such that $|U_I| = k$ and $H(U_I|\mathbf{Y_j}) = k - 1$ for
all $j \in I$. Let $i = \min \{ j: S_j \in U_I \}$.
Then by virtue of Steps 1) and 2) of CAPM, we must have $i \in I$.
Since $H(U_I|\mathbf{Y_i}) = k-1$, we must have $S_{(i+1)} \in U_I$
as well, where $(j) = ((j - 1) \mod m) + 1$. Likewise, $(i+1) \in I$,
which implies that $(i+2) \in I$, etc. It follows, then, that
$|I| = |U_I| = m$.

Conversely, selecting the permutation $\sigma(i) = m - i + 1$
in Theorem~\ref{thm:lower_opt} shows that $R_{DSM+} \le m-1$.
\end{proof}
\begin{figure}[t]
  \begin{center}
 \includegraphics[scale=0.45]{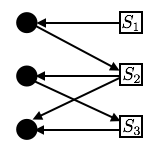}
\caption{Bipartite graph representation of index coding example with $m=3$ and $\mathbf{S} = \{S_1,S_2,S_3\}$. Circle nodes represent users while square nodes denote source bits.}
   \label{fig:bipartite}
 \end{center}
\end{figure}
 We can represent any (groupcast) index coding problem as a bipartite graph 
$G= (\mathbf{M}, \mathbf{S}, E)$ where $\mathbf{M}$, $\mathbf{S}$, and $E$ 
denote the set of user nodes, the set of source bit nodes and the set of 
edges respectively~\cite{neely}. 
There is a directed edge $(m_i, S_i)$, $m_i \in \mathbf{M}$, $S_i \in \mathbf{S}$ if and only if  $m_i$ has  $S_i$ as side information (i.e., $S_i \in \mathbf{Y_{m_i}}$) and there is a directed edge  $(S_i, m_i)$, $S_i \in \mathbf{S}$, $m_i \in \mathbf{M}$  if and only if  $m_i$ demands  $S_i$ (i.e., $S_i \in \mathbf{f_{m_i}}$ and see Figure \ref{fig:bipartite} for an example).
 
 Second, we consider an index coding instance represented by a directed acyclic graph (DAG). Tehrani~\emph{et al.}~\cite{neely} show that partition multicast achieves the optimal rate for DAGs and the optimal rate equals to total number of demanded bits, $s = |\mathbf{S}|$.
Note that for any given instance of an index coding problem, the rate achieved by CAPM cannot be more than the number of demanded bits. Thus it suffices
to show that the rate $s$ is optimal. Tehrani~\emph{et al.} show this
under the zero-error formulation. Using the DSM+ lower bound, one
can show that the optimal rate is also $s$ under the vanishing block
error probability assumption.
\begin{lemma}
\label{lemma:DAG}
For an instance of the  index coding problem represented by a DAG, there exists a permutation, $\sigma(\cdot)$, on $[m]$ such that
\begin{align}
\label{eq:lb_DAG}
\mathbf{Y_{\sigma(i)}} \subseteq \cup^{i-1}_{j=1} \mathbf{f_{\sigma(j)}}, \mbox{ for all } i \in [m].
\end{align} 
\end{lemma}

\begin{proof} 
We follow Neely~\emph{et al.}~\cite{neely_dynamic}.
Observe that every index coding problem represented
by a DAG must have a node in the graph with no outgoing edges.
This node must represent some decoder $\ell$ since every source bit
is assumed to be demanded by at least one decoder
(see Section~\ref{sec:problem_defn_index}). Decoder $\ell$
must then have no side information. Let $\sigma(1) = \ell$.

We then proceed by induction. Suppose the containment in 
(\ref{eq:lb_DAG}) holds for all $i$ in $[k]$ with $k < m$. 
Consider the modified index coding instance in which we 
delete Decoders $\sigma(1), \ldots, \sigma(k)$ and all of
their incoming and outgoing edges in the graph. We also
delete any source components that are left with no edges.
This instance must again be a DAG, and since $k < m$ it must
have at least one decoder node, so there must be a decoder
$\nu$ that has no side information. It follows that in the
original instance, $\mathbf{Y}_{\nu} \subset 
\cup_{j = 1}^{k} \mathbf{f}_{\sigma(j)}.$ We then set $\sigma(k+1) = \nu$.
\end{proof}
By Lemma \ref{lemma:DAG}, we have that the DSM+ lower bound $R_{DSM+}$ is greater than or equal to
\begin{align*}
&H(\mathbf{f_{\sigma(1)}}|\mathbf{Y_{\sigma(1)}})
+ H( \mathbf{f_{\sigma(2)}}|\mathbf{f_{\sigma(1)}},\mathbf{Y_{\sigma(1)}},\mathbf{Y_{\sigma(2)}}) 
+\cdots+
H( \mathbf{f_{\sigma(m)}}|\mathbf{f_{\sigma(1)}},\ldots,\mathbf{f_{\sigma(m-1)}},\mathbf{Y_{\sigma(1)}},\ldots,\mathbf{Y_{\sigma(m)}}) 
\\
&= H(\mathbf{f_{\sigma(1)}})
+ H( \mathbf{f_{\sigma(2)}}|\mathbf{f_{\sigma(1)}}) 
+\cdots+
H( \mathbf{f_{\sigma(m)}}|\mathbf{f_{\sigma(1)}},\ldots,\mathbf{f_{\sigma(m-1)}}), 
\end{align*}
giving $R_{DSM+} \ge s$.
Hence  CAPM gives the optimal rate for DAGs.

Finally, Tehrani~\emph{et al.}~\cite{neely} show that partition
multicast achieves the optimal rate when each decoder demands
a single bit and has as side information all of the other 
source bits. Note that under these assumptions one may,
without loss of generality, assume that each source bit
is demanded by at most one decoder; two decoders
that demand the same source bit must have the same side
information and therefore one of the two can be deleted
without affecting the rate. Then each source component
must be present as side information at either all or
all but one of the decoders.

We shall prove that CAPM is optimal for the
more general scenario in which each source bit is
present at none of the decoders, all of the decoders, all 
but one, or all but two. That is, $\mathbf{S}$ consists of 
$\{ G_0, G_m, G_{m-1}, G_{m-2} \}$. We do not assume that
each decoder demands a single bit or that each bit is
demanded by at most one decoder.

\begin{theorem}
\label{theorem:index}
The optimal rate, $R_{opt}$, for the index coding problem where $\mathbf{S} = \{ G_0, G_m, G_{m-1}, G_{m-2} \}$ is
	\begin{align}
	\label{optrate}
	R_{opt} &= \max\{R_1, \ldots, R_m\}
	\end{align}
	where 
	\begin{align}
	\label{optrate_i}
	R_i =  & H(\mathbf{f_1}\backslash \{\cup_{\{1,\beta\} \subseteq [m]}f_{1\{1,\beta\} }\},\ldots, 
\mathbf{f_{i-1}}\backslash \{\cup_{\{i-1,\beta\} \subseteq [m]}f_{i-1\{i-1,\beta\} }\}, 
	\notag \\
	&\mathbf{f_i},\mathbf{f_{i+1}}\backslash \{\cup_{\{i+1,\beta\} \subseteq [m]}f_{i+1\{i+1,\beta\} }\},\ldots, 
\mathbf{f_m} \backslash \{\cup_{\{m,\beta\} \subseteq [m]}f_{m\{m,\beta\} }\}|\mathbf{Y_i}) +\max_{j, j\in [m]\backslash i}f_{j\{i,j\}},
	\end{align}
and is achieved by CAPM.
\end{theorem}

\begin{corollary}
\label{corr:3decoder}
For any index coding problem with three or fewer decoders, the 
optimal rate is given by (\ref{optrate}), and is achieved
by CAPM.
\end{corollary}

\begin{proof}
Any index coding problem with three or fewer decoders must have
the property that each source component is present as side 
information at either all of the decoders, none of the decoders,
all but one, or all but two.
\end{proof}

\begin{remark}
Since CAPM and the DSM+ lower bound only give integer-valued
bounds, it follows that the optimal rate is integer-valued
for the scenario described in Theorem~\ref{theorem:index},
and in particular, in Corollary~\ref{corr:3decoder}.
\end{remark}

\begin{remark}
Corollary~\ref{corr:3decoder} solves the index coding problem with 
three decoders and any number of source components. In contrast,
Arbabjolfaei~\emph{et al.}~\cite{arbabjolfaei} solve the index coding
problem with up to five source components and any number of decoders.
Evidently neither of these results implies the other, even if one
ignores slight differences in the problem formulation between
the two papers. It is also
worth noting that the Arbabjolfaei~\emph{et al.} result is numerical
while Corollary~\ref{corr:3decoder} is analytical.
\end{remark}

\begin{proof}[Proof of Theorem \ref{theorem:index}]
\textit{\\1. Achievability:}

We utilize the achievable scheme of Proposition \ref{ach_general} and apply the CAPM for selection of $U_{I}$ to get an explicit expression. 
\\
\textit{Step 1}:
 Note that there is no demand related to $G_m$ since all decoders have it as side information.  We place all demands of all decoders related to $G_{m-1}$ and $G_{0}$ into $U_{[m]}$. 
The remaining demands are subsets of $G_{m-2}$, i.e., components that $m-2$ of decoders have. 
For any $\mathbf{S_J} \in G_{m-2}$, where $J=\{\alpha,\beta\}$, $\alpha \neq \beta$ Decoder $\alpha$ and Decoder $\beta$ are the two decoders that do not have  $\mathbf{S_J}$ as side information. Then  we place $f_{\{\alpha,\beta\}J}$ to $U_{[m]}$. 
Also,  we place  $f_{\alpha J}$ and $f_{\beta J}$ to level $m-1$ messages. 
Since there is no demand related to $G_{i}, \forall i \in [m-3]$, all messages $U_{I}, |I| \leq m-2$ will be empty.  This completes the Step $1$.

\textit{Step $2$ and $3$}: To determine the leftover bits in Step $2$ and bits to be XORed in Step $3$ we write the demands in the following way:
Note that there are $m(m-1)$ different non overlapping pairs of demands related to $G_{m-2}$ since $|G_{m-2}| = \frac{m(m-1)}{2}$ and there are two demands $f_{\alpha J}$ and $f_{\beta J}$ for each $\mathbf{S_{\{\alpha,\beta\}}} \in G_{m-2}$. Also, note that for each Decoder $\alpha$ there are $m-1$  non overlapping demands, $f_{\alpha J}$ related to $G_{m-2}$. Therefore, we can put all those non overlapping demands into a matrix $A$ with $m$ rows and $m-1$ columns in the following way. $\alpha^{th}$ row, denoted by $A_\alpha$, consists of demands $f_{\beta\{\alpha,\beta\}}$ where $\beta$ runs over the set $[m]\backslash \{\alpha\}$.
Note that  $A_{\alpha}$ does not contain any demand from the Decoder $\alpha$. Also, for each $f_{\beta\{\alpha,\beta\}}$, all entries of $A_\alpha$ other than $f_{\beta\{\alpha,\beta\}}$ exist as side information at Decoder 
$\beta$.
Hence, we observe that all non overlapping demands which are related to $G_{m-2}$ and placed in $U_{\alpha^\cc}$ at Step $1$  are in $A_{\alpha}$. Also, there is no other type of demand in level $m-1$ messages.

Lastly, as the size of each demand  in $A_\alpha$ can be different, we arrange the entries in $A_\alpha$ in an increasing order with respect to their sizes. If two demands are in equal size, which one is put first does not matter. This completes the construction of the matrix $A$.  

For each $A_\alpha$ in $A$ we apply the following $\oplus$ operation. 
 \begin{definition} 
\label{def5}
Let $a_i, i\in \{1,\ldots,m-1\}$ be  vectors. Assume without loss of generality that $l_1 \le l_2\le,\ldots,\le l_{m-1}$ where $l_i = |a_i|$ denotes the number of elements in $a_i$. Then, 
\begin{align*}
(a_1,a_2,\ldots,a_{m-1})^{\oplus} = (a^\oplus_1,\ldots,a^\oplus_{m-1}) 
\end{align*}
where 
\begin{align*}
(a^\oplus_1,\ldots,a^\oplus_{m-1})  =( &a_1\oplus(a_2)_{l_1}\oplus\cdots\oplus(a_{m-1})_{l_1}, \\
&(a_2)_{l_2-l_1}\oplus\cdots\oplus(a_{m-1})_{l_2-l_1}, \\
& \cdots, \\
&(a_{m-1})_{l_{m-1}-l_{m-2}}), \\
\end{align*} 
and where $(a_i)_{l_j-l_k}$ denotes the vector consisting of components of $a_i$  from $(l_k +1)^{th}$ to $l_j^{th}$ component.
 \end{definition}
Note that all the components in $A^\oplus_{\alpha2},\ldots,A^\oplus_{\alpha(m-1)}$ are leftover bits and moved to $U_{[m]}$.
Then for each $A_\alpha^{\oplus}=(A^\oplus_{\alpha1},\ldots,A^\oplus_{\alpha(m-1)})$, all the components in $A^\oplus_{\alpha1}$ (\textit{i.e.,} $ A_{\alpha1},(A_{\alpha2})_{|A_{\alpha1}|},\ldots,(A_{\alpha(m-1)})_{|A_{\alpha1}|}$) remain in $U_{\alpha^\cc}$.
This concludes Step 2.

Lastly by Step $3$, we have $U_{\alpha^\cc}= A^\oplus_{\alpha1}$, for all $\alpha \in [m]$, 
 and $A^\oplus_{\alpha2},\ldots,A^\oplus_{\alpha(m-1)}$ are in $U_{[m]}$. This concludes CAPM.
Also for the ease of notation, when we write $A^\oplus\backslash \{\cup_{\alpha \in [m]}A^\oplus_{\alpha1} \}$, we will mean the vector $A^\oplus$ with the components $A^\oplus_{\alpha1}$ for all $\alpha \in [m]$ removed.

 Hence,
\begin{align*}
&U_{[m]} = \mathbf{f_1}\backslash \{\cup_{\{1,\beta\} \subseteq [m]}f_{1\{1,\beta\}} \}, \mathbf{f_2}\backslash \{\cup_{\{2,\beta\} \subseteq [m]}f_{2\{2,\beta\} }\},
 \ldots,\mathbf{f_m}\backslash \{\cup_{\{m,\beta\} \subseteq [m]}f_{m\{m,\beta\} } \}, A^\oplus\backslash \{\cup_{\alpha \in [m]}A^\oplus_{\alpha1} \}
 \\
& U_{\alpha^\cc}  = A^\oplus_{\alpha1} \quad \forall \alpha \in [m],
\end{align*}
and  we can write the achievable rate $R_{CAPM}$ as
\begin{align}
&  \max_{k \in [m]} \{ H(\mathbf{f_1}\backslash \{\cup_{\{1,\beta\} \subseteq [m]}f_{1\{1,\beta\}} \}, \mathbf{f_2}\backslash \{\cup_{\{2,\beta\} \subseteq [m]}f_{2\{2,\beta\} }\}
 ,\ldots,\mathbf{f_m}\backslash \{\cup_{\{m,\beta\} \subseteq [m]}f_{m\{m,\beta\} } \}, A^\oplus\backslash \{\cup_{\alpha \in [m]}A^\oplus_{\alpha1} \} |\mathbf{Y_k})\}  
\notag \\
\label{ach rate}
&\quad + \max_{k\in \{1\}^\cc}\{H( A^\oplus_{11}|\mathbf{Y_k})\} + \cdots 
+ \max_{k\in \{m\}^\cc}\{H( A^\oplus_{m1}|\mathbf{Y_k})\}. 
\end{align}

Note that 
\begin{align}
\label{exor}
&H(A_{\alpha1},(A_{\alpha2})_{|A_{\alpha1}|},\ldots,(A_{\alpha(m-1)})_{|A_{\alpha1}|} |\mathbf{Y_k}) \\
\label{noexor}
&\overset{a}{=} H(A_{\alpha1}\oplus(A_{\alpha2})_{|A_{\alpha1}|}\oplus,\cdots,\oplus(A_{\alpha(m-1)})_{|A_{\alpha1}|} |\mathbf{Y_k}) \\
&=H( A^\oplus_{\alpha1}|\mathbf{Y_k}).
\end{align}
Before rearranging the terms in $R_{CAPM}$ further, we would like to make the following remarks.
\begin{remark}
\label{remark1}
From the Definition~\ref{def5}, we know that for all $[m] \setminus \{\alpha\}$ the conditional entropy
\begin{align}
& H( A^\oplus_{\alpha1}|\mathbf{Y_k})
\notag \\
&=H((A_{\alpha j})_{|A_{\alpha1}|} |\mathbf{Y_k}) \notag \\
&= H((A_{\alpha j})_{|A_{\alpha1}|}) = |A_{\alpha1}| =\min_{l\in [m-1]}|A_{\alpha l}| \mbox{ bits,} \notag
\end{align}
 where $A_{\alpha j}$  is the demand at Decoder $k$ related to $G_{m-2}$.
\end{remark}
Hence we get,
\begin{align*}
H( U_{\alpha^\cc}|\mathbf{Y_k}) &=H( U_{\alpha^\cc}|\mathbf{Y_j})
\\
& =\min_{l\in [m-1]}|A_{\alpha l}| \mbox{ bits}, \forall k,j \in \{\alpha\}^\cc
\end{align*}
 
\begin{remark}
\label{remark2}
By Definition~\ref{def5}, for all $\alpha \in [m]$,
\begin{align*}
H(A^\oplus\backslash A^\oplus_{\alpha1} |\mathbf{Y_\alpha}) &=H(A^\oplus\backslash A^\oplus_{\alpha},(A^\oplus_{\alpha2},\ldots,A^\oplus_{\alpha(m-1)}) |\mathbf{Y_\alpha}) 
\\
&=H(A^\oplus\backslash A^\oplus_{\alpha}|\mathbf{Y_\alpha}) 
 +H(A^\oplus_{\alpha2},\ldots,A^\oplus_{\alpha(m-1)} |\mathbf{Y_\alpha}) 
\\
&=H(\cup_{\{\alpha,\beta\} \subseteq [m]}f_{\alpha \{\alpha,\beta\}} |\mathbf{Y_\alpha}) 
+ H(A^\oplus_{\alpha2},\ldots,A^\oplus_{\alpha(m-1)})  \\
&=H(\cup_{\{\alpha,\beta\} \subseteq [m]}f_{\alpha \{\alpha,\beta\}} |\mathbf{Y_\alpha}) 
+ (\max_{j}|A_{\alpha j}|- \min_{j}|A_{\alpha j}|).
\end{align*}
\end{remark}

%

By Remark~\ref{remark1}, $H( U_{\alpha^\cc}|\mathbf{Y_k}) =H( U_{\alpha^\cc}|\mathbf{Y_j}) $ for all $k,j \in \{\alpha\}^\cc$. When we expand the terms inside $\max_{k\in [m]}$, we can write
\begin{align*}
R_{CAPM}=&\max \Big\{ H(\mathbf{f_1}\backslash \{\cup_{\{1,\beta\} \subseteq [m]}f_{1\{1,\beta\}} \},\ldots,
 \mathbf{f_m}\backslash \{\cup_{\{m,\beta\} \subseteq [m]}f_{m\{m,\beta\} } \},
A^\oplus\backslash \{\cup_{\alpha \in [m]}A^\oplus_{\alpha1} \}|\mathbf{Y_1}) 
 \\
&\quad \quad \quad+ \sum_{\alpha\in \{1\}^\cc} H( U_{\alpha^\cc}|\mathbf{Y_1}) + 
\max_{k\in \{1\}^\cc}\{H(U_{1^\cc}|\mathbf{Y_k})\}, 
 \\
& \quad \quad H(\mathbf{f_1}\backslash \{\cup_{\{1,\beta\} \subseteq [m]}f_{1\{1,\beta\}} \},\ldots,
\mathbf{f_m}\backslash \{\cup_{\{m,\beta\} \subseteq [m]}f_{m\{m,\beta\} } \}, 
A^\oplus\backslash \{\cup_{\alpha \in [m]}A^\oplus_{\alpha1} \}|\mathbf{Y_2})  
\\
&\quad \quad+ \sum_{\alpha\in \{2\}^\cc} H( U_{\alpha^\cc}|\mathbf{Y_2}) + 
\max_{k\in \{2\}^\cc}\{H(U_{2^\cc}|\mathbf{Y_k})\}
\\
&\quad \quad,\ldots, 
\\
&\quad \quad  H(\mathbf{f_1}\backslash \{\cup_{\{1,\beta\} \subseteq [m]}f_{1\{1,\beta\}} \},  \ldots,
 \mathbf{f_m}\backslash \{\cup_{\{m,\beta\} \subseteq [m]}f_{m\{m,\beta\} } \}, 
A^\oplus\backslash \{\cup_{\alpha \in [m]}A^\oplus_{\alpha1} \}|\mathbf{Y_m}) 
\\
&\quad \quad + \sum_{\alpha\in \{m\}^\cc} H(U_{\alpha^\cc}|\mathbf{Y_m}) + 
\max_{k\in \{m\}^\cc}\{H(U_{m^\cc}|\mathbf{Y_k})\} \Big\}.
\end{align*}

Since $H( U_{\alpha^\cc}|\mathbf{Y_k}) =H( A^\oplus_{\alpha1}|\mathbf{Y_k})$  and $H( U_{\alpha^\cc}|\mathbf{Y_k}) = |A^\oplus_{\alpha 1}|$  for all $k \in \{\alpha\}^\cc$ by Remark ~\ref{remark1}, we can further write $R_{CAPM}$ as
\begin{align*}
&\max \Big\{ H(\mathbf{f_1}\backslash \{\cup_{\{1,\beta\} \subseteq [m]}f_{1\{1,\beta\}} \},\ldots,
\mathbf{f_m}\backslash \{\cup_{\{m,\beta\} \subseteq [m]}f_{m\{m,\beta\} } \}, 
A^\oplus\backslash \{\cup_{\alpha \in [m]}A^\oplus_{\alpha1} \}|\mathbf{Y_1})  
\\
&\quad \quad \quad+ \sum_{\alpha\in \{1\}^\cc} H( A^\oplus_{\alpha1}|\mathbf{Y_1}) + 
|A^\oplus_{11} |, 
 \\
&\quad \quad H(\mathbf{f_1}\backslash \{\cup_{\{1,\beta\} \subseteq [m]}f_{1\{1,\beta\}} \},\ldots, 
 \mathbf{f_m}\backslash \{\cup_{\{m,\beta\} \subseteq [m]}f_{m\{m,\beta\} } \},
A^\oplus\backslash \{\cup_{\alpha \in [m]}A^\oplus_{\alpha1} \}|\mathbf{Y_2}) 
 \\
&\quad \quad + \sum_{\alpha\in \{2\}^\cc} H( A^\oplus_{\alpha1}|\mathbf{Y_2}) + 
|A^\oplus_{21} | 
\\
&\quad \quad ,\ldots, \\
&\quad \quad H(\mathbf{f_1}\backslash \{\cup_{\{1,\beta\} \subseteq [m]}f_{1\{1,\beta\}} \},\ldots, 
\mathbf{f_m}\backslash \{\cup_{\{m,\beta\} \subseteq [m]}f_{m\{m,\beta\} } \}, 
A^\oplus\backslash \{\cup_{\alpha \in [m]}A^\oplus_{\alpha1} \}|\mathbf{Y_m}) 
\\
&\quad \quad + \sum_{\alpha\in \{m\}^\cc} H(A^\oplus_{\alpha1}|\mathbf{Y_m}) + 
|A^\oplus_{m1} | \Big\}.
\end{align*}

By applying chain rule to the expression above and since all 
$U_{i^\cc} = A^\oplus_{i1}$ satisfy the condition $2$ in Proposition \ref{ach_general}, we have
\begin{align*}
R_{CAPM}=
& \max \Big\{ H(\mathbf{f_1}\backslash \{\cup_{\{1,\beta\} \subseteq [m]}f_{1\{1,\beta\}} \}, 
\ldots, \mathbf{f_m}\backslash \{\cup_{\{m,\beta\} \subseteq [m]}f_{m\{m,\beta\} } \}, 
A^\oplus\backslash \{\cup_{\alpha \in [m]}A^\oplus_{\alpha1} \}|\mathbf{Y_1})  
\\
&\quad \quad \quad+  H( \cup_{\alpha\in \{1\}^\cc} A^\oplus_{\alpha1}|\mathbf{Y_1}) + 
|A^\oplus_{11} |, 
 \\
&\quad \quad H(\mathbf{f_1}\backslash \{\cup_{\{1,\beta\} \subseteq [m]}f_{1\{1,\beta\}} \}, 
\ldots, \mathbf{f_m}\backslash \{\cup_{\{m,\beta\} \subseteq [m]}f_{m\{m,\beta\} } \}, A^\oplus\backslash \{\cup_{\alpha \in [m]}A^\oplus_{\alpha1} \}|\mathbf{Y_2})  
\\
&\quad \quad +  H( \cup_{\alpha\in \{2\}^\cc}A^\oplus_{\alpha1}|\mathbf{Y_2}) + 
|A^\oplus_{21} | 
\\
&\quad \quad ,\ldots, \\
&\quad \quad  H(\mathbf{f_1}\backslash \{\cup_{\{1,\beta\} \subseteq [m]}f_{1\{1,\beta\}} \},\ldots,
 \mathbf{f_m}\backslash \{\cup_{\{m,\beta\} \subseteq [m]}f_{m\{m,\beta\} } \},
A^\oplus\backslash \{\cup_{\alpha \in [m]}A^\oplus_{\alpha1} \}|\mathbf{Y_m}) 
\\
 &\quad \quad + H(\cup_{\alpha\in \{m\}^\cc} A^\oplus_{\alpha1}|\mathbf{Y_m}) + 
|A^\oplus_{m1} | \Big\}.
\end{align*}

Applying chain rule  once more and since all 
$U_{i^\cc} = A^\oplus_{i1}$ satisfy the condition $2$ in Proposition \ref{ach_general}, we get 
\begin{align*}
R_{CAPM} =
&\max \Big\{ H(\mathbf{f_1}\backslash \{\cup_{\{1,\beta\} \subseteq [m]}f_{1\{1,\beta\}} \}, 
\ldots, \mathbf{f_m}\backslash \{\cup_{\{m,\beta\} \subseteq [m]}f_{m\{m,\beta\} } \}|\mathbf{Y_1}) 
+ H(A^\oplus\backslash A^\oplus_{11} |\mathbf{Y_1})  
+ |A^\oplus_{11} |,  
\\
&\quad \quad  H(\mathbf{f_1}\backslash \{\cup_{\{1,\beta\} \subseteq [m]}f_{1\{1,\beta\}} \},\ldots,
 \mathbf{f_m}\backslash \{\cup_{\{m,\beta\} \subseteq [m]}f_{m\{m,\beta\} } \} |\mathbf{Y_2}) +H(A^\oplus\backslash A^\oplus_{21} |\mathbf{Y_2})  
+ |A^\oplus_{21} |  \\
&\quad \quad,\ldots, \\
&\quad \quad H(\mathbf{f_1}\backslash \{\cup_{\{1,\beta\} \subseteq [m]}f_{1\{1,\beta\}} \},\ldots,
 \mathbf{f_m}\backslash \{\cup_{\{m,\beta\} \subseteq [m]}f_{m\{m,\beta\} } \} |\mathbf{Y_m}) 
+ H(A^\oplus\backslash A^\oplus_{m1} |\mathbf{Y_m}) 
+ |A^\oplus_{m1} | \Big\}.
\end{align*}

Finally, using Remark \ref{remark2}, that $|A_{\alpha1}|=\min_{j}|A_{\alpha j}|$ and chain rule, we have the following expression for $R_{CAPM}$.
\begin{align*}
R_{CAPM} =
&\max \Big\{ H(\mathbf{f_1}, \mathbf{f_2}\backslash \{\cup_{\{2,\beta\} \subseteq [m]}f_{2\{2,\beta\} }\},\ldots,
 \mathbf{f_m}\backslash \{\cup_{\{m,\beta\} \subseteq [m]}f_{m\{m,\beta\} } \} |\mathbf{Y_1})  +\max_{j}|A_{1j}|,  
\\
&\quad H(\mathbf{f_1}\backslash \{\cup_{\{1,\beta\} \subseteq [m]}f_{1\{1,\beta\} }\}, \mathbf{f_2},\mathbf{f_3}\backslash \{\cup_{\{3,\beta\} \subseteq [m]}f_{3\{3,\beta\} }\},
\ldots,
 \mathbf{f_m}\backslash \{\cup_{\{m,\beta\} \subseteq [m]}f_{m\{m,\beta\} } \} |\mathbf{Y_2})  +\max_{j}|A_{2j}| 
 \\
&\quad,\ldots, 
\\
&\quad H(\mathbf{f_1}\backslash \{\cup_{\{1,\beta\} \subseteq [m]}f_{1\{1,\beta\} }\},\ldots,
 \mathbf{f_{m-1}}\backslash \{\cup_{\{m-1,\beta\} \subseteq [m]}f_{m-1\{m-1,\beta\} }\} ,\mathbf{f_m} |\mathbf{Y_m}) 
+\max_{j}|A_{mj}| \Big\}.
\end{align*} 

Then we can write this achievable rate for the problem, $R_{CAPM}$ in (\ref{ach rate}) as
\begin{equation}
R_{CAPM} = \max \{{R_{1},\ldots,R_{m}} \}. 
\end{equation}

Next, we find a lower bound which matches $R_{CAPM}$  above by utilizing the converse result in section \ref{sec:lower_bound_index}.
Let us focus on (\ref{rlo1}).
Here, we can write 
\begin{align}
&H(\mathbf{f_2}|\mathbf{Y_2},\mathbf{f_1},\mathbf{Y_1})
\notag \\
 &= H(\mathbf{f_2}|\mathbf{f_1},\mathbf{Y_1}) \notag \\
&= H(\mathbf{f_2}\backslash\{\cup_{\{2,\beta\} \in [m]}f_{2\{2,\beta\}} \},\{\cup_{\{2,\beta\} \in [m]}f_{2\{2,\beta\}} \}|\mathbf{f_1},\mathbf{Y_1}) \notag \\
&= H(\mathbf{f_2}\backslash\{\cup_{\{2,\beta\} \in [m]}f_{2\{2,\beta\}} \},f_{2\{2,1\}} \}|\mathbf{f_1},\mathbf{Y_1}) \notag \\
&=  H(\mathbf{f_2}\backslash\{\cup_{\{2,\beta\} \in [m]}f_{2\{2,\beta\}} \}, A_{1j}|\mathbf{f_1},\mathbf{Y_1}), \notag \\
 &\text{ where } A_{1j}=f_{2\{2,1\}}\notag \\
&=  H(\mathbf{f_2}\backslash\{\cup_{\{2,\beta\} \in [m]}f_{2\{2,\beta\}} \}|\mathbf{f_1},\mathbf{Y_1}) +H(A_{1j}) \notag \\
\label{entropyf2}
&=  H(\mathbf{f_2}\backslash\{\cup_{\{2,\beta\} \in [m]}f_{2\{2,\beta\}} \}|\mathbf{f_1},\mathbf{Y_1}) +|A_{1j}|
\end{align}
and 
\begin{align}
&H(\mathbf{f_3}|\mathbf{Y_3},\mathbf{f_2},\mathbf{Y_2},\mathbf{f_1},\mathbf{Y_1})
\notag \\
  &= H(\mathbf{f_3}|\mathbf{f_2},\mathbf{Y_2},\mathbf{f_1},\mathbf{Y_1}) \notag \\
&= H(\mathbf{f_3}\backslash\{\cup_{\{3,\beta\} \in [m]}f_{3\{3,\beta\}} \},\{\cup_{\{3,\beta\} \in [m]}f_{3\{3,\beta\}} \} |\mathbf{f_2},\mathbf{Y_2},\mathbf{f_1},\mathbf{Y_1})  \\
&=  H(\mathbf{f_3}\backslash\{\cup_{\{3,\beta\} \in [m]}f_{3\{3,\beta\}} \}|\mathbf{f_2},\mathbf{Y_2},\mathbf{f_1},\mathbf{Y_1}), \notag \\
&\text{since } \{\cup_{\{3,\beta\} \in [m]}f_{3\{3,\beta\}} \} \subset  \{\mathbf{Y_1},\mathbf{Y_2}\}  \notag \\
\label{entropyf3}
&=  H(\mathbf{f_3}\backslash\{\cup_{\{3,\beta\} \in [m]}f_{3\{3,\beta\}} \}|\mathbf{f_2}\backslash\{\cup_{\{2,\beta\} \in [m]}f_{2\{2,\beta\}} \}, \mathbf{Y_2},\mathbf{f_1},\mathbf{Y_1}), \\ 
&\text{since } \{\cup_{\{2,\beta\} \in [m]}f_{2\{2,\beta\}} \} \perp  \mathbf{f_3}\backslash\{\cup_{\{3,\beta\} \in [m]}f_{3\{3,\beta\}} \} | \mathbf{Y_i},\mathbf{f_j} \mbox{ for all } i,j \in [m]. \notag
\end{align}

Note that $\mathbf{f_3}$ can be written as 
\begin{align*}
\{f_{3[m]},f_{33},\cup_{\{3,\beta\} \in [m]}f_{3\{3,\beta\}},\cup_{\{3,\beta\} \in [m]}f_{\{3,\beta\}\{3,\beta\}} \}.
\end{align*}

Then we get the following equality:
\begin{align*}
  &\mathbf{f_3}\backslash \{\cup_{\{3,\beta\} \in [m]}f_{3\{3,\beta\}} \}
=\{f_{3[m]} ,f_{33},\cup_{\{3,\beta\} \in [m]}f_{\{3,\beta\}\{3,\beta\}}\}.
\end{align*}

Note that $f_{\{3,1\}\{3,1\}} \subseteq \mathbf{f_1}$ and $f_{\{3,\beta\}\{3,\beta\}}
\subseteq \mathbf{Y_1}$, for all $\beta \in [m]\backslash\{3\}$ and $\beta \neq 1$.
 Also, $f_{3[m]} \notin \mathbf{Y_2}$, $f_{33} \subseteq \mathbf{Y_1}$,
$f_{33} \subseteq \mathbf{Y_2}$. As a result, (\ref{entropyf3}) can be written as
\begin{align*}
&H(\mathbf{f_3}\backslash\{\cup_{\{3,\beta\} \in [m]}f_{3\{3,\beta\}} \}|\mathbf{f_2}\backslash\{\cup_{\{2,\beta\} \in [m]}f_{2\{2,\beta\}} \},\mathbf{f_1},\mathbf{Y_1}).
\end{align*}

By similar arguments as above, for $p > 2$ we can write
\begin{align}
&H(\mathbf{f_p}|\mathbf{Y_p},\mathbf{f_{p-1}},\mathbf{Y_{p-1}},\ldots,\mathbf{f_1},\mathbf{Y_1})
\notag \\
  & = H(\mathbf{f_p}|\mathbf{f_{p-1}},\mathbf{Y_{p-1}},\ldots,\mathbf{f_1},\mathbf{Y_1}) \notag \\
&= H(\mathbf{f_p}\backslash\{\cup_{\{p,\beta\} \in [m]}f_{p\{p,\beta\}} \},\cup_{\{p,\beta\} \in [m]}f_{p\{p,\beta\}} |\mathbf{f_{p-1}},
\mathbf{Y_{p-1}},\ldots,\mathbf{f_1},\mathbf{Y_1}) \notag \\
&\overset{a}{=}  H(\mathbf{f_p}\backslash\{\cup_{\{p,\beta\} \in [m]}f_{p\{p,\beta\}} \}|\mathbf{f_{p-1}},\mathbf{Y_{p-1}},\dots,\mathbf{f_1},\mathbf{Y_1}) \notag \\
&\overset{b}{=}  H(\mathbf{f_p}\backslash\{\cup_{\{p,\beta\} \in [m]}f_{p\{p,\beta\}} \} 
|\mathbf{f_{p-1}}\backslash\{\cup_{\{p-1,\beta\} \in [m]}f_{p-1\{p-1,\beta\}} \},\mathbf{Y_{p-1}},\ldots,
\mathbf{f_{2}} \backslash\{\cup_{\{2,\beta\} \in [m]}f_{2\{2,\beta\}} \},\mathbf{Y_{2}},\mathbf{f_1},\mathbf{Y_1}) \notag \\
&\overset{c}{=}  H(\mathbf{f_p}\backslash\{\cup_{\{p,\beta\} \in [m]}f_{p\{p,\beta\}} \} 
|\mathbf{f_{p-1}}\backslash\{\cup_{\{p-1,\beta\} \in [m]}f_{p-1\{p-1,\beta\}} \},\ldots,
\label{entropyfp}
\mathbf{f_{2}} \backslash\{\cup_{\{2,\beta\} \in [m]}f_{2\{2,\beta\}} \},\mathbf{f_1},\mathbf{Y_1})  
\end{align} 

a: Since  $\{\cup_{\{p,\beta\} \in [m]}f_{p\{p,\beta\}} \} \subset  \{\mathbf{Y_i},\mathbf{Y_j}\}$, $ \forall i,j \in [m], p\neq i, j \text{ and } i\neq j.$ 

b: Since $\{\cup_{\{\alpha,\beta\} \in [m]}f_{\alpha\{\alpha,\beta\}} \} \perp   \mathbf{f_\gamma}\backslash\{\cup_{\{\gamma,\beta\} \in [m]}f_{\gamma\{\gamma,\beta\}} \}$ $\ \ \ \ | \mathbf{Y_i},\mathbf{f_j}$, $\forall i,j,\alpha,\gamma \in [m]. $

c: Since $\mathbf{f_p}=\{f_{p[m]} ,f_{pp} ,\{\cup_{\{p,\beta\} \in [m]}f_{p\{p,\beta\}} \}$,
$\{\cup_{\{p,\beta\} \in [m]}f_{\{p,\beta\}\{p,\beta\}} \},\}$ and
 $\mathbf{f_p}\backslash \{\cup_{\{p,\beta\} \in [m]}f_{p\{p,\beta\}} \}$
 \\
 equals
$\{f_{p[m]} ,f_{pp} ,\{\cup_{\{p,\beta\} \in [m]}f_{\{p,\beta\}\{p,\beta\}} \} \}$,
 where
 \\
 $f_{p[m]} \notin \mathbf{Y_i}, \forall i\in[m]$ and $f_{pp} \subseteq \mathbf{Y_i}$  $\forall i \in [m], i \neq q$. Also, $f_{\{p,1\}\{p,1\}} \subseteq \mathbf{f_1}$ and 
 $f_{\{p,\beta\}\{p,\beta\}} \subseteq \mathbf{Y_1}$, $\forall \beta \in [m]\backslash \{p\}$, $\beta \neq 1$.

Hence, from $(\ref{entropyf2})$ and $(\ref{entropyfp})$, $(\ref{rlo1})$ can be written as
\begin{align}
R\ge & H(\mathbf{f_1},\mathbf{f_2}\backslash \{\cup_{\{2,\beta\} \in [m]}f_{2\{2,\beta\}} \},\ldots, 
\label{Rlo1}
\mathbf{f_m}\backslash \{\cup_{\{m,\beta\} \in [m]}f_{m\{m,\beta\}} \}|\mathbf{Y_1}) +|A_{1j}|, 
\end{align} 
 where $ A_{1j}=f_{2\{2,1\}}$.

 By similar arguments used to obtain $(\ref{Rlo1})$, we can write 
\begin{align}
R\ge &  H(\mathbf{f_1},\mathbf{f_2}\backslash \{\cup_{\{2,\beta\} \in [m]}f_{2\{2,\beta\}} \},\ldots, 
\label{Rlo1final}
\mathbf{f_m}\backslash \{\cup_{\{m,\beta\} \in [m]}f_{m\{m,\beta\}} \}|\mathbf{Y_1})  + \max_{j}|A_{1j}|.
\end{align} 

Note that the right hand side of $(\ref{Rlo1final})$ is $R_{1}$. Since the problem is symmetric, similarly we can get all $R_{i}$'s. Then,
 \begin{align}
R\ge & \max_i{R_{i}} =R_{CAPM}
\end{align} 
proving that $R_{CAPM}$ is optimal. This concludes the proof of the Theorem \ref{theorem:index}.

\end{proof}

The following result illustrates the importance of excess bits.

\begin{proposition}
\label{ach_step1}
If the demands of the $m$-user index coding problem are such that there are no excess bits after Step $1$ of CAPM then the rate obtained by following only Step $1$ is optimal. The optimal rate $R_*$ can be written as
\begin{align}
R_* &= \max_{\pi} \{  H(\mathbf{f_{\pi(1)}}|\mathbf{Y_{\pi(1)}}) + H(\mathbf{f_{\pi(2)}}|\mathbf{f_{\pi(1)}},\mathbf{Y_{\pi(1)}}, \mathbf{Y_{\pi(2)}}) 
\notag \\
&\quad + \cdots+
\notag \\
\label{eq: opt_step1}
 & \quad H(\mathbf{f_{\pi(m)}}|\mathbf{f_{\pi(1)}},\mathbf{Y_{\pi(1)}},\ldots,\mathbf{f_{\pi(m-1)}},\mathbf{Y_{\pi(m-1)}},\mathbf{Y_{\pi(m)}}) \}
\end{align}
where $\pi(.)$ denotes the following $m$ permutations on $[m]$:
\begin{align*}
(1,2,\ldots,m), (2,1,3, \ldots,m), \ldots, (m,1,\ldots, m-1).
\end{align*}
\end{proposition}

\begin{proof}

First we show that the achievable rate we get by applying Step $1$ of CAPM gives the expression in (\ref{eq: opt_step1}). We begin with the following three observations. Firstly, all demands of each Decoder $i$, $\mathbf{f_i}$, are in $\cup_{i \in I} U_{I}$. Secondly, since the demands are such that there are no excess bits after Step $1$, $H(U_{I} | \mathbf{Y_i}) = H(U_{I} | \mathbf{Y_j})$, for all $i,j \in I \subset [m]$. 
Lastly, demands placed in $U_{[m]}$ at Step $1$ cannot be excess bits since $U_{[m]}$ is the highest level message. Hence $H(U_{[m]} | \mathbf{Y_i})$ does not have to be equal for all $i \in [m]$.

We can write the achievable rate $\RC$ as
\begin{align*}
\RC =\max\{ R_{1},\ldots, R_{m}\}, \textrm{where}
\end{align*}
\begin{align}
R_{i} &=H(U_{[m]} | \mathbf{Y_i} ) + \sum_{I \subset [m]}\max_{j \in I}  \{H(U_{I}| \mathbf{Y_j})\}
\notag \\
\label{eqn:ach_noexcess}
&\overset{a}{=} H(U_{[m]} | \mathbf{Y_i} ) + \sum_{I \subset [m]} H(U_{I}| \mathbf{Y_{i_{I}}}), 
\end{align}
 where $i_{I}$ is an arbitrary element of $I$ and $a$ is due to the assumption that
  $H(U_{I} | \mathbf{Y_i}) = H(U_{I} | \mathbf{Y_j})$, for all $i,j \in I \subset [m]$.

Let us focus on $R_{1}$. From (\ref{eqn:ach_noexcess}), we can write $R_{1}$ as

\begin{align}
R_{1} &= \sum_{C_1} H(U_{I}| \mathbf{Y_1}) + \sum_{C_2} H(U_{I}| \mathbf{Y_2}) + \cdots 
+ \sum_{C_m} H(U_{I}| \mathbf{Y_m})  
\notag 
\end{align}
where $C_1 = \{I \subseteq [m]|1 \in I\}$, $C_2=\{I \subseteq [m] | 2 \in I, 1 \not\in I\}$, $\ldots$,$C_m =\{ I \subseteq [m]|  m \in I, 1 \notin I,\ldots, m-1 \not\in I \}$.
\\
 Since  all $U_I$'s satisfy the condition 2) in Proposition \ref{ach_general}, $R_{1}$ equals
\begin{align}
& H( \cup_{C_1}U_{I}| \mathbf{Y_1}) + H(\cup_{C_2} U_{I}| \mathbf{Y_2}) + \cdots 
 + H(\cup_{C_m} U_{I}| \mathbf{Y_m}) 
\notag \\ 
&{=}  H( \mathbf{U_1}| \mathbf{Y_1}) + H(\mathbf{U_2} \setminus\mathbf{U_1}| \mathbf{Y_2}) + \cdots 
+ H(\mathbf{U_m}\setminus\{ \mathbf{U_1},\ldots,\mathbf{U_{m-1}}\}| \mathbf{Y_m}) \notag,
\end{align} 
where $\mathbf{U_i}$ is defined as $\cup_{I \subseteq [m]:i \in I}U_{I}$.
By Step $1$, no decoder in $I^c$ can demand any source bit in $U_{I}$ or 
have it as side information. Then we can write $R_1$ as 
\begin{align}
& R_1 =H( \mathbf{U_1}| \mathbf{Y_1}) + H(\mathbf{U_2} \setminus\mathbf{U_1}| \mathbf{Y_1}, \mathbf{Y_2}) + \cdots 
\label{ineq:example}
 + H(\mathbf{U_m}\setminus\{ \mathbf{U_1},\ldots,\mathbf{U_{m-1}}\}| \mathbf{Y_1},\ldots,\mathbf{Y_m}). 
\end{align}

Also, by Step $1$, each $U_{I}$ consists of those source bits such
that, for each decoder $i$ in $I$, Decoder $i$ either demands the
bit or has it as side information. Then (\ref{ineq:example}) becomes
\begin{align}
& H( \mathbf{f_1}| \mathbf{Y_1}) + H(\mathbf{f_2} \setminus \mathbf{f_1}| \mathbf{Y_1}, \mathbf{Y_2}) + \cdots 
+ H(\mathbf{f_m}\setminus\{\mathbf{f_1},\ldots,\mathbf{f_{m-1}}\}| \mathbf{Y_1},\ldots,\mathbf{Y_m}) 
\notag \\
\label{R_ach1}
&= H( \mathbf{f_1}| \mathbf{Y_1}) + H(\mathbf{f_2}  | \mathbf{f_1},\mathbf{Y_1}, \mathbf{Y_2}) + \cdots 
 + H(\mathbf{f_m}| \mathbf{f_1},\ldots,\mathbf{f_{m-1}},\mathbf{Y_1},\ldots,\mathbf{Y_m}).
\end{align}

Note that the expression for $R_{1}$ in (\ref{R_ach1}) is equivalent to first expression of the $R_*$. 
Applying the procedure above to the other $R_{i}$'s similarly, we see that  $\RC$ gives  the expression in (\ref{eq: opt_step1}). Evidently this
expression cannot exceed the lower bound in Theorem \ref{thm:lower_opt},
so the proof is complete.
\end{proof}

The \emph{coded caching} problem, which was introduced by Maddah-Ali and Niesen~\cite{maddah_cache}, is closely related to the index coding problem.
The coded caching problem consists of two phases, called 
the \emph{cache allocation} phase and the \emph{delivery} phase.
During the cache allocation phase, the server can decide how to populate
the caches of the various users. Each user then selects some
content to demand, and during the delivery phase the server must broadcast
a common message to all of the clients that allows each one to
meet its demand, given its cache contents. Thus the delivery
phase of the coded caching problem can be viewed as an index coding
problem.

If we perform the cache allocation as in \cite{Maddah-Ali}  and each user demands a different file at the delivery phase, then the instance of the index coding problem that results during the delivery phase satisfies the conditions in Proposition \ref{ach_step1} in a certain asymptotic sense. Therefore, CAPM gives the optimal rate for the delivery phase in this case.

\section{S-CAPM: A Heuristic Achieving Fractional Rates}
\label{sec:second_scheme}
Recall that CAPM can only give integer rates. However, some instances of the index coding problem are known to have  
non-integer optimal rates. We next show how CAPM can be modified
to give nonintegral rate bounds, and this modification performs
strictly better than CAPM in some examples. The extension is
not polynomial-time computable, however.
The following multi-letter extension of Proposition~\ref{ach_general2} is necessary.
\begin{proposition} 
\label{ach_general2} 
Let $t$ be a positive integer. The optimal rate $R_{opt}$ of an index coding problem is upper bounded by
\begin{align*}
\min \frac{1}{t}\sum_{I\subseteq [m]}  \left[ \max_{
i\in I}  H(U^t_{I} |\mathbf{{Y^t_i}}) \right]
\end{align*}
where the minimization is over the set of all random variables $U^t_{I}$ jointly distributed with $\mathbf{S^t}$ such that
\\
1) There exist functions

 $g_1(\cup_{1 \in I}U^t_I,\mathbf{Y^t_1})$,$\ldots$, $g_m(\cup_{m \in I}U^t_I,\mathbf{Y^t_m})$ 
  such that
\begin{align*}
& g_i(\cup_{i \in I}U^t_I,\mathbf{Y^t_i}) = \mathbf{f^t_i}(\mathbf{S}), \mbox{ for all } i \in [m].
\end{align*}
2) The auxiliary random variables $U^t_I$, $I \subseteq [m]$ are independent,
and for all collections of subsets $J_1, \ldots J_j$, $K_1, \ldots K_k$,
$L_1, \ldots, L_l$, and all subsets $\{i_1, \ldots, i_p\} \subseteq [m]$,
we have that
$(U^t_{J_1}, \ldots, U^t_{J_j})$ and $(U^t_{K_1}, \ldots, U^t_{K_k})$ are
conditionally independent given $((U^t_{L_1}, \ldots, U^t_{L_l}), 
(\mathbf{Y^t_{i_1},\ldots, Y^t_{i_p}})$, provided that the collections
 $J_1, \ldots J_j$ and  $K_1, \ldots K_k$ are disjoint. \\
3) Each $U^t_{I}$ is a (possibly empty) vector of bits, each of which
is the mod-2 sum of a set (possibly singleton) of source components.
\end{proposition}

Proposition~\ref{ach_general} can evidently be recovered from 
Proposition~\ref{ach_general2}
by taking $t = 1$. But Proposition~\ref{ach_general} actually implies
Proposition~\ref{ach_general2}, since the latter can be obtained
by applying the former to blocks of size $t$.

Now we provide a heuristic, which we call \emph{Split Coded Approximate Partition Multicast (S-CAPM)}, for selecting the auxiliary random variables in 
Proposition~\ref{ach_general2}.
The steps for S-CAPM are very similar to the ones for CAPM in Section~\ref{sec:first_scheme} except for the placement of leftover bits.

\textbf{Step 1} (Initialization) : This step is exactly the same as in CAPM,
except that we shall parametrize the solution differently. For each 
$k \in \{1,\ldots,|\mathbf{S}|\}$ and each subset $I \subseteq [m]$, let
$\theta(I,k)$ denote a variable in the interval $[0,1]$. We shall
interpret $\theta(I,k)$ as the ``fraction'' of source
bit $S_k$ that is allocated
to the auxiliary random variable $U_I$.  All such variables
are initially zero. 

For each source component $k$ we set $\theta(K \cup J^c,k) = 1$, where
$K$ and $J$ are chosen so that $S_k$ is in $f_{KJ}$. This is assuming
that $J \ne [m]$. As in CAPM, if $J = [m]$ then we set $\theta([m],k) = 1$.
Note that after this has been done for each $k$, we have
$$
\sum_{I \subset [m]} \theta(I,k) = 1
$$
for each $k$. This equality will remain true after Step 2.

\textbf{Step 2} : As with CAPM, the goal of Step 2 is to promote
``excess bits'' to a higher-level message. Since each auxiliary
random variable now stores fractional bits, however, both the
notion of ``excess'' and the promotion process are more involved.

Given the variables $\{\theta(I,k)\}$, let us define the ``conditional
entropy'' of $U_I$ given $\mathbf{Y_j}$ as
\begin{equation}
\label{eq:fracentropy}
H(U_I|\mathbf{Y_j}) = \sum_{k: S_k \notin \mathbf{Y_j}} \theta(I,k).
\end{equation}

Note that if $\theta(I,k) \in \{0,1\}$ for all $I$ and $k$, then
this reduces to the conditional entropy examined in Step 2 of CAPM.
We shall be most interested in $H(U_I|\mathbf{Y_j})$ when $j \in I$,
although the definition in (\ref{eq:fracentropy}) does not require this.

We then perform the following procedure for each subset $I$. The order
in which we process the different subsets $I$ is not specified by the
heuristic, except that if $|I_1| < |I_2|$ then $I_1$ must be processed
prior to $I_2$. For a given subset $I$, we define
\begin{equation}
\label{eq:step2tmin}
i^* = \min\{i: H(U_I|\mathbf{Y_i}) = \min_{l\in I} H(U_I|\mathbf{Y_l})\}
\end{equation}
and
\begin{equation}
\label{eq:step2tmax}
j^* = \min\{j: H(U_I|\mathbf{Y_j}) = \max_{l \in I} H(U_I|\mathbf{Y_l})\}.
\end{equation}

If $H(U_I|\mathbf{Y_{i^*}}) = H(U_I|\mathbf{Y_{j^*}})$ the we are done with this 
subset and may move to the next one. If $H(U_I|\mathbf{Y_{i^*}}) < H(U_I|\mathbf{Y_{j^*}})$,
then let $E$ denote the set of source bits that are ``excess''
$$
E = \{k: \theta(I,k) > 0 \ \text{and} \ S_k \in \mathbf{Y_{i^*}} \ \text{but} \
           S_k \notin \mathbf{Y_{j^*}}\}.
$$

We then select a source bit in $E$ to promote to higher-level
messages. Consider the set
\begin{equation}
\label{eq:excess}
\{k \in E: \theta(I,k) \le H(U_I|\mathbf{Y_{j^*}}) - H(U_I|\mathbf{Y_{i^*}})\}.
\end{equation}

If this set is nonempty, then there is at least one source bit
that is ``entirely excess.'' We shall select one such bit to
promote. Choose an arbitrary
\begin{align*}
k^* & \in \arg\max \{\theta(I,k): k \in E \ \text{and }  
             \theta(I,k) \le H(U_I|\mathbf{Y_j^*}) - H(U_I|\mathbf{Y_i^*})\}.
\end{align*}

We then set $\theta(I,k^*) = 0$ and we increment $\theta(I',k^*)$
for all $I'$ such that $I \subseteq I'$ and $|I'| = |I| + 1$
by the amount
$$
\frac{\theta(I,k^*)}{|I^c|}.
$$

In words, we view $\theta(I,k^*)$ as an amount of fluid that
is removed from $U_I$ and divided equally among the $I^c$ sets $I'$.

If there are no bits that are entirely excess, i.e., the set
in (\ref{eq:excess}) is empty, then choose an arbitrary
$$
k^* \in \arg\min \{\theta(I,k): k \in E\}.
$$

We then promote only the portion of $\theta(I,k)$ that is excess.
That is, we replace $\theta(I,k)$ with
$$
H(U_I|\mathbf{Y_j^*}) - H(U_I|\mathbf{Y_i^*})
$$
and divide the remaining part,
$$
\theta(I,k) - (H(U_I|\mathbf{Y_j^*}) - H(U_I|\mathbf{Y_i^*}))
$$
equally among all of the sets $I'$ such that $I \subset I'$ and
$|I'| = |I| + 1$. Observe that since $\theta(I,k)$ must be
rational for all $I$ and $k$, the process will eventually
terminate.

\textbf{Step 3} : As in CAPM, we now look for opportunities to exclusive-OR
source bits included in the same auxiliary random variable. First we 
convert the fractional bits described by the $\theta(\cdot,\cdot)$
variables to an integral number by increasing the parameter $t$. Observe
that $\theta(I,k)$ must be rational for each $I$ and $k$; let $t$ denote
the smallest positive integer so that $\theta(I,k) \cdot t$ is an integer
for all $I$ and $k$. Next recall that for each $k$
$$
\sum_{I \subseteq [m]} \theta(I,k) \cdot t = t.
$$

We then divide the block of $t$ bits corresponding to source component
$k$ among the $U_I$ variables so that the number of bits that $U_I$
receives is $\theta(I,k) \cdot t$. One can verify that the resulting
$U_I$ variables satisfy conditions 1)-3) in Proposition~\ref{ach_general2}. 
For each $U_I$ variable, we then look for exclusive-OR opportunities
as in Step 3 of CAPM, resulting in revised $U_I$ variables that remain
feasible.

We next illustrate S-CAPM with two examples.

\label{ex1}
\begin{figure}[t]
  \begin{center}
 \includegraphics[scale=0.3]{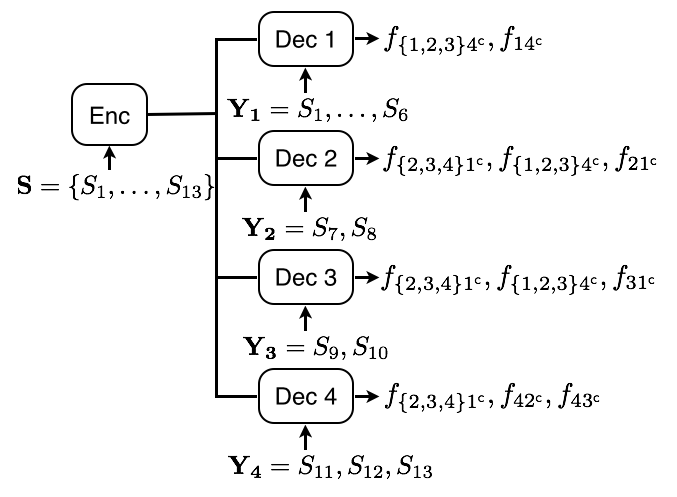}
\caption{Index coding example with $4$ users}
   \label{fig:ex1}
 \end{center}
\end{figure}

\begin{example}
In this case, there are $4$ decoders with side information and demands as shown in Fig. \ref{fig:ex1}, where 
 $f_{\{2,3,4\}1^\cc} = (S_1,S_2)$, $f_{21^\cc} = (S_3,S_4)$,
 $f_{31^\cc} =( S_5,S_6)$, $f_{42^\cc}=(S_7,S_8)$, $f_{43^\cc}=(S_9,S_{10})$, $f_{\{1,2,3\}4^\cc}=(S_{11},S_{12})$ and $f_{14^\cc}=S_{13}$. 
By using S-CAPM, we determine the messages and $t$ of the achievable scheme. 

\textit{Step 1 :} At the end of this step, all of the following $\theta(I,k)$'s are unity:
\\
 $\theta(\{1,4\},13)$, $\theta(\{1,2\},3),\theta(\{1,2\},4)$,
$\theta(\{1,3\}, 5), \theta(\{1,3\},6)$, $\theta(\{2,4\},7), \theta(\{2,4\},8)$, $\theta(\{3,4\},9), \theta(\{3,4\},10)$,
\\
$\theta([4],1), \theta([4],2),\theta([4],11),\theta([4],12)$.

\textit{Step 2}: 
We start with level-$2$ messages. Note that all demands in 
level-$2$ messages are excess bits. Since there are two possible level-$3$ messages that each demand can move, we set  all the corresponding $\theta(I,k)$s to $0.5$. At this point the nonzero $\theta(I,k)$'s  are
\begin{align*}
&\theta(\{1,2,4\},13), \theta(\{1,2,4\},7), \theta(\{1,2,4\},8), \theta(\{1,2,4\},3),
\\
& \theta(\{1,2,4\},4), 
\\
&\theta(\{1,3,4\},13), \theta(\{1,3,4\},9), \theta(\{1,3,4\},10), 
\theta(\{1,3,4\}, 5),
\\
&  \theta(\{1,3,4\},6),
\\
&\theta(\{1,2,3\}, 5), \theta(\{1,2,3\},6), \theta(\{1,2,3\},3),
\theta(\{1,2,3\},4), 
\\
&\theta(\{2,3,4\},7), \theta(\{2,3,4\},8), \theta(\{2,3,4\},9), \theta(\{2,3,4\},10), 
\\
&\theta([4],1), \theta([4],2),
 \theta([4],11),\theta([4],12),
\end{align*}
where $\theta(I,k) =0.5$ for all $|I|=3$ and $\theta(I,k) = 1$ for all $|I|=4$.
Now we move on to level-$3$ messages. 
Since there is only one level-$4$ message, $U_{1234}$, all possible excess bits at this stage will be moved to $U_{1234}$. We start with $U_{124}$.
Since $H(U_{124} | \mathbf{Y_1}) = 1.5$, 
$H(U_{124} | \mathbf{Y_2}) = 1.5$, $H(U_{124} | \mathbf{Y_4}) = 2$, we have $i^* =1$, $j^*= 4$. We declare, say, $S_3$ to be excess and we move
all of $\theta(\{1,2,4\},3)$ to $U_{1234}$.
Then we recalculate $H(U_{124} | \mathbf{Y_i})$, for $i \in \{1,2,4\}$. Now $i^* =2$, $j^*= 1$ and the fraction of $S_7$, i.e., $\theta(\{1,2,4\},7)$, becomes an excess bit.
We recalculate $H(U_{124} | \mathbf{Y_i})$, for $i \in \{1,2,4\}$ and all are equal. Hence we move on to another level-$3$ message, say  $U_{134}$. For this message fraction of $S_5$ and $S_9$ become excess bits and are moved to $U_{1234}$. Lastly,
all demands in $U_{123}$, $U_{234}$ are excess bits and moved to $U_{1234}$. This concludes Step $2$. 
\\
\textit{Step 3:} Since there is no $XOR$ opportunities as described in this step, we only require $t$ to be $2$.
Then the nonzero $\theta(I,k)$'s are
\begin{align*}
&\theta(\{1,2,4\},13), \theta(\{1,2,4\},8),\theta(\{1,2,4\},4),
\\
& \theta(\{1,3,4\},13), \theta(\{1,3,4\},10),\theta(\{1,3,4\},6),
\\
& \theta([4],1), \theta([4],2),\theta([4],11),
\theta([4],12), \theta([4],3), \theta([4],7), 
\\
&\theta([4],5), \theta([4],9),
 \theta([4],6),
 \theta([4],4), 
 \theta([4],8),  \theta([4],10),
\end{align*} 
where $\theta(I,k) =0.5$ for  $k \in \{4,6,8,10,13\} $ and the rest are $1$.
 As a result,  the rate coming from  level-$3$ and level-$4$ messages are  $2$ and  $8.5$ bits respectively and the total rate for this problem is $10.5$ bits.
 
From the linear programming lower bound\footnote{This lower bound is for the zero-error setting. However, it can be modified to handle vanishing block error probabilities. } stated in \cite{blasiak}, we get $10.5$ bits showing that S-CAPM is optimal.
\end{example}

\label{ex2}
\begin{figure}[t]
  \begin{center}
 \includegraphics[scale=0.35]{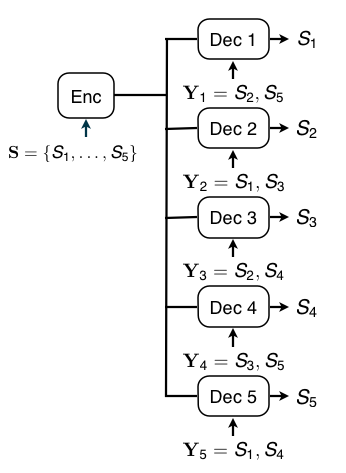}
\caption{Index coding example with $5$ users}
   \label{fig:ex2}
 \end{center}
\end{figure}
\begin{example}
We consider the ``$5$-cycle'' index coding problem shown in Fig.~\ref{fig:ex2}. Its optimal rate is found in  \cite{blasiak}.  When we apply S-CAPM, we determine the messages and $t$ as follows:

Step 1 : After this step we have the following nonzero $\theta(I,k)$'s.
\\
 $\theta(\{1,2,5\},1)$, $\theta(\{1,2,3\},2)$, $\theta(\{2,3,4\},3)$, $\theta(\{3,4,5\},4)$, $\theta(\{1,4,5\},5)$,
  where all $\theta(I,k) = 1$ for all $k \in [5]$.
 
Step $2$ : We start with level-$3$ messages. Note that all demands in level-$3$ messages are excess bits to be moved to level-$4$ messages. Since  each leftover bit has two possible level-$4$ messages to go,  $\theta(I,k) = 0.5$ for all $k \in [5]$. Then the nonzero $\theta(I,k)$'s are
\\
$\theta(\{1,2,3,5\},1), \theta(\{1,2,3,5\},2)$, 
\\
$ \theta(\{1,2,4,5\},1), \theta(\{1,2,4,5\},5)$, 
\\
$\theta(\{1,2,3,4\},2), \theta(\{1,2,3,4\},3)$, 
\\
$\theta(\{1,3,4,5\},4), \theta(\{1,3,4,5\},5)$, 
\\
$\theta(\{2,3,4,5\},3), \theta(\{2,3,4,5\},4)$.

Now, we move on to level-$4$ messages.
Since there are no leftover bits at level-$4$ messages and all nonzero $\theta(I,k) = 0.5$, we set $t=2$ concluding S-CAPM. 
Hence, the total rate becomes $2.5$ bits which is the optimal rate.
\end{example}
\section{ACKNOWLEDGMENT}
This work was supported by Intel, Cisco, and
Verizon under the Video-Aware Wireless Networks (VAWN) program.

\bibliographystyle{IEEEtran}
\bibliography{IEEEabrv,Index_Coding_Final}

\begin{thebibliography}{10}
\providecommand{\url}[1]{#1}
\csname url@samestyle\endcsname
\providecommand{\newblock}{\relax}
\providecommand{\bibinfo}[2]{#2}
\providecommand{\BIBentrySTDinterwordspacing}{\spaceskip=0pt\relax}
\providecommand{\BIBentryALTinterwordstretchfactor}{4}
\providecommand{\BIBentryALTinterwordspacing}{\spaceskip=\fontdimen2\font plus
\BIBentryALTinterwordstretchfactor\fontdimen3\font minus
  \fontdimen4\font\relax}
\providecommand{\BIBforeignlanguage}[2]{{%
\expandafter\ifx\csname l@#1\endcsname\relax
\typeout{** WARNING: IEEEtran.bst: No hyphenation pattern has been}%
\typeout{** loaded for the language `#1'. Using the pattern for}%
\typeout{** the default language instead.}%
\else
\language=\csname l@#1\endcsname
\fi
#2}}
\providecommand{\BIBdecl}{\relax}
\BIBdecl

\bibitem{birk1}
Y.~Birk and T.~Kol, ``Informed-source coding-on-demand (iscod) over broadcast
  channels,'' in \emph{INFOCOM '98. Seventeenth Annual Joint Conference of the
  IEEE Computer and Communications Societies. Proceedings. IEEE}, vol.~3, 1998,
  pp. 1257--1264 vol.3.

\bibitem{birk}
Z.~Bar-Yossef, Y.~Birk, T.~S. Jayram, and T.~Kol, ``Index coding with side
  information,'' in \emph{Foundations of Computer Science, 2006. FOCS '06. 47th
  Annual IEEE Symposium on}, 2006, pp. 197--206.

\bibitem{shanmugam_local}
K.~Shanmugam, A.~Dimakis, and M.~Langberg, ``Local graph coloring and index
  coding,'' in \emph{Information Theory Proceedings (ISIT), 2013 IEEE
  International Symposium on}, July 2013, pp. 1152--1156.

\bibitem{shanmugam}
K.~Shanmugam, A.~G. Dimakis, and M.~Langberg, ``Graph theory versus minimum
  rank for index coding,'' in \emph{Information Theory (ISIT), 2014 IEEE
  International Symposium on}, June 2014, pp. 291--295.

\bibitem{blasiak}
A.~Blasiak, R.~Kleinberg, and E.~Lubetzky, ``Broadcasting with side
  information: Bounding and approximating the broadcast rate,''
  \emph{Information Theory, IEEE Transactions on}, vol.~59, no.~9, pp.
  5811--5823, 2013.

\bibitem{neely}
A.~Tehrani, A.~Dimakis, and M.~Neely, ``Bipartite index coding,'' in
  \emph{Information Theory Proceedings (ISIT), 2012 IEEE International
  Symposium on}, July 2012, pp. 2246--2250.

\bibitem{Maleki}
H.~Maleki, V.~Cadambe, and S.~Jafar, ``Index coding ;an interference alignment
  perspective,'' \emph{Information Theory, IEEE Transactions on}, vol.~60,
  no.~9, pp. 5402--5432, Sept 2014.

\bibitem{arbabjolfaei}
F.~Arbabjolfaei, B.~Bandemer, Y.-H. Kim, E.~Sasoglu, and L.~Wang, ``On the
  capacity region for index coding,'' in \emph{Information Theory Proceedings
  (ISIT), 2013 IEEE International Symposium on}, July 2013, pp. 962--966.

\bibitem{arbab_14}
F.~Arbabjolfaei and Y.-H. Kim, ``Local time sharing for index coding,'' in
  \emph{Information Theory (ISIT), 2014 IEEE International Symposium on}, June
  2014, pp. 286--290.

\bibitem{block}
N.~Alon, E.~Lubetzky, U.~Stav, A.~Weinstein, and A.~Hassidim, ``Broadcasting
  with side information,'' in \emph{Foundations of Computer Science, 2008. FOCS
  '08. IEEE 49th Annual IEEE Symposium on}, 2008, pp. 823--832.

\bibitem{nonlinear}
E.~Lubetzky and U.~Stav, ``Nonlinear index coding outperforming the linear
  optimum,'' \emph{Information Theory, IEEE Transactions on}, vol.~55, no.~8,
  pp. 3544--3551, 2009.

\bibitem{berger}
C.~Heegard and T.~Berger, ``Rate distortion when side information may be
  absent,'' \emph{Information Theory, IEEE Transactions on}, vol.~31, no.~6,
  pp. 727--734, 1985.

\bibitem{Kaspi:Absent}
A.~H. Kaspi, ``Rate-distortion function when side-information may be present at
  the decoder,'' \emph{Information Theory, IEEE Transactions on}, vol.~40,
  no.~6, pp. 2031--2034, 1994.

\bibitem{timo}
R.~Timo, T.~Chan, and A.~Grant, ``Rate distortion with side-information at many
  decoders,'' \emph{Information Theory, IEEE Transactions on}, vol.~57, no.~8,
  pp. 5240--5257, 2011.

\bibitem{ElGamal}
A.~E. Gamal and Y.-H. Kim, \emph{Network Information Theory}.\hskip 1em plus
  0.5em minus 0.4em\relax New York, NY, USA: Cambridge University Press, 2012.

\bibitem{Arikan:packet}
E.~Arikan, ``Some complexity results about packet radio networks,''
  \emph{Information Theory, IEEE Transactions on}, vol.~30, no.~4, pp.
  681--685, 1984.

\bibitem{vishwanath}
S.~Vishwanath, G.~Kramer, S.~Shamai, S.~Jafar, and A.~Goldsmith, ``Capacity
  bounds for {G}aussian vector broadcast channels,'' \emph{DIMACS SERIES IN
  DISCRETE MATHEMATICS AND THEORETICAL COMPUTER SCIENCE}, vol.~62, pp.
  107--122, 2004.

\bibitem{vishwanathP}
P.~Vishwanath and D.~Tse, ``On the capacity of the multiple antenna broadcast
  channel,'' \emph{DIMACS SERIES IN DISCRETE MATHEMATICS AND THEORETICAL
  COMPUTER SCIENCE}, vol.~62, pp. 87--106, 2004.

\bibitem{watanabe}
S.~Watanabe, ``The rate-distortion function for product of two sources with
  side-information at decoders,'' \emph{Information Theory, IEEE Transactions
  on}, vol.~59, no.~9, pp. 5678--5691, Sept 2013.

\bibitem{timo:complementary}
R.~Timo, A.~Grant, and G.~Kramer, ``Lossy broadcasting with complementary side
  information,'' \emph{Information Theory, IEEE Transactions on}, vol.~59,
  no.~1, pp. 104--131, 2013.

\bibitem{sinem_gauss}
S.~Unal and A.~B. Wagner, ``Vector {G}aussian rate-distortion with variable
  side information,'' in \emph{Information Theory Proceedings (ISIT), 2014 IEEE
  International Symposium on}, 2014.

\bibitem{sgarro}
A.~Sgarro, ``Source coding with side information at several decoders,''
  \emph{Information Theory, IEEE Transactions on}, vol.~23, no.~2, pp.
  179--182, 1977.

\bibitem{timo:gray}
R.~Timo, A.~Grant, T.~Chan, and G.~Kramer, ``Source coding for a simple network
  with receiver side information,'' in \emph{Information Theory Proceedings
  (ISIT), 2008 IEEE International Symposium on}, 2008, pp. 2307--2311.

\bibitem{timo:conditional}
R.~Timo, T.~J. Oechtering, and M.~Wigger, ``Source coding problems with
  conditionally less noisy side information,'' \emph{Information Theory, IEEE
  Transactions on}, vol.~60, no.~9, pp. 5516--5532, 2014.

\bibitem{laich}
T.~Laich and M.~Wigger, ``Utility of encoder side information for the lossless
  {K}aspi/{H}eegard-{B}erger problem,'' in \emph{Information Theory Proceedings
  (ISIT), 2013 IEEE International Symposium on}, 2013, pp. 3065--3069.

\bibitem{langberg}
M.~Langberg and M.~Effros, ``Network coding: Is zero error always possible?''
  in \emph{Communication, Control, and Computing (Allerton), 2011 49th Annual
  Allerton Conference on}, Sept 2011, pp. 1478--1485.

\bibitem{sinemITA}
S.~Unal and A.~B. Wagner, ``A rate-distortion approach to index coding,'' in
  \emph{ITA}.\hskip 1em plus 0.5em minus 0.4em\relax IEEE, 2014, pp. 1--5.

\bibitem{neely_dynamic}
M.~Neely, A.~Tehrani, and Z.~Zhang, ``Dynamic index coding for wireless
  broadcast networks,'' \emph{Information Theory, IEEE Transactions on},
  vol.~59, no.~11, pp. 7525--7540, Nov 2013.

\bibitem{maddah_cache}
M.~Maddah-Ali and U.~Niesen, ``Fundamental limits of caching,''
  \emph{Information Theory, IEEE Transactions on}, vol.~60, no.~5, pp.
  2856--2867, May 2014.

\bibitem{Maddah-Ali}
M.~A. Maddah-Ali and U.~Niesen, ``Decentralized caching attains order-optimal
  memory-rate tradeoff,'' \emph{CoRR}, vol. abs/1301.5848, 2013.

\bibitem{Rockafellar:Convex}
R.~T. Rockafellar, \emph{Convex Analysis}.\hskip 1em plus 0.5em minus
  0.4em\relax Princeton University Press, 1970.

\end{thebibliography}

\appendices
\section*{Appendix}
\begin{lemma}
	\label{lemma:cardinality}
     The minimum in~(\ref{eq:Umin}) is unaffected by the presence of 
     the cardinality bounds in condition 5).
\end{lemma}
	
	\begin{proof}
	Without loss of generality let $\sigma(i) = i$ for all $i \in [m]$.
	Let $P_{U_1,\ldots, U_m,{S}}(u_1,\ldots, u_m,s)$ 
	denote the joint distribution of $(U_1,\ldots, U_m,{S})$. 
	We follow a procedure similar to that of \cite{watanabe}.
	First we find the bound on the cardinality of $U_1$, then
    $U_2$, etc. 
	
	To begin with, we consider the following  $(|\mathcal{S}| -1) + 1 + m$ functions of $P_{U_2,\ldots,U_m,{S}|U_1}(.,\ldots,.|u_1)$, denoted as
	 $g^0_{s}, s \in |\mathcal{S}|-1$ and $g^0_{lo}, g^0_{d_1},\ldots,g^0_{d_m}$.
	\begin{align}
	\label{car:pmf}
	&g^0_{s}(P_{U_2,\ldots, U_m,{S}|U_1}(.,\ldots,.|u_1)) 
= \sum_{u_2,\ldots, u_m}P_{U_2,\ldots,U_m,{S}|U_1}(u_2,\ldots, u_m,s|u_1),
	\end{align}
	for $s = 1,\ldots, |\mathcal{S}| -1$ and
	\begin{align}
	&g^0_{lo} (P_{U_2,\ldots,U_m,{S}|U_1}(.,\ldots,.|u_1)) 
	 \notag \\
	&=	H({S}|{Y_1})- H({S}| U_1=u_1, {Y_1}) 
	+ I( {S}; U_{2}|	U_{1}= u_1, {Y_{1}},{Y_{2}}) +\cdots
	 \label{eq:car_exp}
+ I({S}; U_{m}| U_{1}=u_1, \ldots, U_{(m-1)},{Y_{1}},\ldots, {Y_{m}}),
	\end{align}
	 and
	 \begin{align*}
	& g^0_{d_1} (P_{U_2,\ldots, U_m,{S}|U_1}(.,\ldots,.|u_1)) 
= \mathbb{E}	[d({S},g_1(u_1,{Y}_1))|U_1=u_1]
	\\
	&\vdots
	\\
	& g^0_{d_m} (P_{U_2,\ldots, U_m,{S}|U_1}(.,\ldots,.|u_1)) 
= \mathbb{E}	[d({S},g_m(U_m,{Y}_m))|U_1=u_1]
	 \end{align*}
	 
	 Then by Carath\'{e}odory's
theorem~\cite[Theorem~17.1]{Rockafellar:Convex}
	 we can find a random variable $U^1_1$ with $|U^1_1| \leq |\mathcal{S}| + m + 1$ and random variables $U^1_2, \ldots,U^1_m$ 	where $P_{U^1_1,\ldots,U^1_m,{S}}(u_1,\ldots,u_m,s)$ $= P_{U_1^1}(u_1)P_{U_2,\ldots,U_m,{S}|U_1}(u_2,\ldots,u_m,s|u_1)$ 	such that from (\ref{car:pmf})  $P_{{S}}$ is preserved  and from (\ref{eq:car_exp}) 
	 \begin{align*}
	&I({S};U^1_1|{Y_1}) + I( {S}; U^1_{2}|U^1_{1},{Y_{1}},{Y_{2}}) +\cdots
+ I( {S}; U^1_{m}|U^1_{1},\ldots,U^1_{(m-1)},{Y_{1}},\ldots,	{Y_{m}}) 
	\\
	&= I({S};U_1|{Y_1}) + I( {S}; U_{2}| U_{1}, {Y_{1}},{Y_{2}}) 
+ I( {S}; U_{m}| U_{1},\ldots, U_{(m-1)},{Y_{1}},\ldots,{Y_{m}}),
	\end{align*}
	and we have
	 \begin{align*}
	& \mathbb{E}[d({S},g_1(U^1_1,{Y}_1)] = \mathbb{E}	[d({S},g_1(U_1,{Y}_1))]
	\\
	&\vdots
	\\
	& \mathbb{E}[d({S},g_m(U^1_m, {Y}_1)] = \mathbb{E}	[d({S},g_m(U_m,{Y}_1))].
	 \end{align*}
	 
Now we consider the following $|U_1||\mathcal{S}| + (m-1)$  functions of
$P_{U_1^1,U^1_3,\ldots,U^1_m,{S}|U^1_2}(.,\ldots,.| u_2)$.
	\begin{align}
	&g^1_{s}(P_{U_1^1,U^1_3,\ldots,U^1_m,{S}|U^1_1}(.,\ldots,.| u_2)) 
	\label{eq:car_pmf2}
	= \sum_{u_3,\ldots,u_m}	P_{U_1^1,U^1_3,\ldots,U^1_m,{S}|U^1_2}(.,	\ldots,.| u_2),
	\end{align}
	for $(u_1,s) = 1,\ldots, |{U_1}||\mathcal{S}| -1$ and
	\begin{align}
	&g^1_{lo} (P_{U_1^1,U^1_3,\ldots,U^1_m,{S}|U^1_2}(.,\ldots,.|u_2)) 	
	\notag \\
	&=  - H( {S}| U^1_{1}, U^1_{2}=u_2,{Y_{1}},{Y_{2}}) 	
 + I( {S}; U^1_{3}| U^1_{1}, U^1_2=u_2,Y_{1}, Y_2,Y_3)
	\cdots
	 \notag \\
	 \label{eq:car_exp2}
	&\quad + I( {S}; U^1_{m}| U^1_{1}, U^1_2=u_2,U^1_3\ldots,U^1_{(m-1)},{Y_{1}},\ldots, {Y_{m}}),
	\end{align}
	 and
	 \begin{align*}
	& g^1_{d_2} (P_{U_1^1,U^1_3,\ldots,U^1_m,{S}|U^1_2}(.,\ldots,.|u_2)) 
= \mathbb{E}[d({S},g_2(u_2,{Y}_2)| U^1_2=u_2]
	\\
	\vdots
	\\
	& g^1_{d_m} (P_{U_1^1,U^1_3,\ldots,U^1_m,{S}|U^1_2}(.,\ldots,.|u_2)) 
= \mathbb{E}[d({S},g_m(U^1_m,{Y}_m)| U^1_2=u_2].
	 \end{align*}
	 
	Again by Carath\'{e}odory's
theorem, there is a random variable $U^2_2$ with 
	$|U^2_2| \leq |U_1||\mathcal{S}| + m$ and random variables $U^2_3,\ldots,U^2_m$ where 
	$P_{U_1^1, U^2_2,\ldots,U^2_m,{S}}(u_1,\ldots,u_m,s)$ is equal to 
	$P_{U^2_2}(u_2)P_{U_1^1,U^1_3,\ldots,U^1_m,{S}|U^1_2}(u_1,u_3,\ldots,u_m,s|u_2)$ 	such that $P_{U_1^1{S}}$ is preserved (from \ref{eq:car_pmf2}). 
	
	Since $P_{U_1^1{S}}$ is preserved,  $\mathbb{E}[d({S},g_2(U^1_1,{Y}_1)]$, $H( {S}| U^1_{1},{Y_{1}},{Y_{2}})$, and $I({S};U_1^1|{Y_1})$ are preserved. Also, from (\ref{eq:car_exp2}) we have
	 \begin{align*}
	&I({S};U_1^1|{Y_1}) + I( {S}; U^2_{2}|U^1_{1}, {Y_{1}},{Y_{2}}) +\cdots
+ I( {S}; U^2_{m}|U_1^1,U^2_{2},\ldots,U^2_{(m-1)},{Y_{1}},\ldots,{Y_{m}}) 
	\\
	&= I({S};U^1_1|{Y_1}) + I( {S}; U^1_{2}| U^1_{1},{Y_{1}},{Y_{2}}) 
 + I( {S}; U^1_{m}| U^1_{1},\ldots,U^1_{(m-1)},{Y_{1}},\ldots,	{Y_{m}}). 
	\end{align*}
	
	Lastly, we have the following equalities.
	 \begin{align*}
	& \mathbb{E}[d({S},g_2(U^2_2,{Y}_2)] = \mathbb{E}[d({S},g_2(U^1_2,{Y}_2))]
	\\
	&\vdots
	\\
	& \mathbb{E}[d({S},g_m(U^1_m,{Y}_m)] 
	 = \mathbb{E}[d({S},g_m(U^1_m{Y}_m))].
	 \end{align*} 
	 
By applying the above procedure to $U^2_3,\ldots,U^2_m$ consecutively and relabeling $(U_1^1,U^2_2,\ldots)$ as $(U_1,\ldots,U_m)$ we obtain the cardinality bounds as stated in the condition 5) of Theorem \ref{lower_general}.
	\end{proof}
	
\end{document}